\documentclass[journal,draftclsnofoot,onecolumn,12pt,twoside]{IEEEtran}
\usepackage{graphicx}
\usepackage{amsmath,amssymb}
\usepackage{cases}
\usepackage{color}
\usepackage{amsfonts}
\usepackage{indentfirst}
\usepackage{slashbox}
\usepackage{cite}
\usepackage[justification=centering]{caption}
\usepackage{diagbox}
\usepackage{bm}
\usepackage{amssymb}
\usepackage{caption}
\usepackage{algorithm}
\usepackage{algpseudocode}
\usepackage{booktabs}
\usepackage{subfigure}
\usepackage{multirow}
\usepackage{amsthm}
\usepackage{chngpage}
\usepackage{epstopdf}
\usepackage{hyperref}
\usepackage{blkarray}
\makeatletter
\newtheoremstyle{mythm}{3pt}{3pt}{}{16pt}{\bfseries}{:}{.5em}{}
\theoremstyle{mythm}
\newtheorem{theorem}{Theorem}
\newtheorem{example}{Example}
\newtheorem{definition}{Definition}
\newtheorem{remark}{Remark}

\newtheorem{proposition}{Proposition}

\newtheorem{lemma}{Lemma}
\newtheorem{construction}{Construction}

\begin{document}
\title{Multi-access Coded Caching with Optimal Rate and Linear Subpacketization under PDA and Consecutive Cyclic Placement
\author{Jinyu~Wang, Minquan~Cheng, and Youlong~Wu}
\thanks{J. Wang and M. Cheng are with Guangxi Key Lab of Multi-source Information Mining $\&$ Security, Guangxi Normal University,
Guilin 541004, China  (e-mail: mathwjy@163.com, chengqinshi@hotmail.com). J. Wang is also with College of Mathematics and Statistics, Guangxi Normal University, Guilin 541004, China.}
\thanks{Youlong Wu is with the School of Information Science and Technology, ShanghaiTech University, 201210 Shanghai, China.
(e-mail: wuyl1@shanghaitech.edu.cn).}
}

\date{}
\maketitle

\begin{abstract}
This work considers the multi-access caching system proposed by Hachem et al., where each user has access to $L$ neighboring caches in a cyclic wrap-around fashion. We first propose a placement strategy called the consecutive cyclic placement, which achieves the maximal local caching gain. Then under the consecutive cyclic placement, we derive the optimal coded caching gain from the perspective of Placement Delivery Array (PDA), thus obtaining a lower bound on the rate of PDA. Finally, under the consecutive cyclic placement, we construct a class of PDA, leading to a multi-access coded caching scheme with linear subpacketization, which achieves our derived lower bound for some parameters; while for other parameters, the achieved coded caching gain is only $1$ less than the optimal one.  Analytical and numerical comparisons of the proposed scheme with existing schemes are provided to validate the performance.

\end{abstract}

\section{Introduction}
\label{intro}
Wireless network has been imposed a tremendous pressure on the network traffic during peak hours, due in large part to multi-media applications such as Video-on-Demand.
Furthermore, the high temporal variability of network traffic results in congestion during peak hours and underutilization of the network during off-peak hours.
Caching has been proposed as an effective method to shift traffic from peak to off-peak hours by placing popular contents into caches across the network during off-peak hours \cite{BBD}. Then the users are partially satisfied from local caches, so the network traffic can be reduced.

The first coded caching scheme was proposed by Maddah-Ali and Niesen (MN) in \cite{MN} for a $(K,M,N)$ centralized caching system, where a central server with $N$ files of unit size is connected by $K$ users, each one has a distinct cache of size $M$ units.
A coded caching scheme operates in two phases. In the placement phase, each file is divided into $F$ equal packets, and some packets are populated into each user's cache without knowing future demands of all users. The quantity $F$ is referred to as the {\em subpacketization}. In the delivery phase, each user requests a file from the server. According to the users' requests as well as the users' cache contents, the server broadcasts coded messages to all users through an error-free shared link, such that each user can recover its requested file.
The goal is to minimize the {\it transmission rate} (or rate) $R$, which is the worst-case normalized transmission amount over all possible demands.
The MN scheme uses an uncoded placement (i.e., some bits of files are directly copied to each user's cache) and a coded delivery policy, such that each message broadcasted by the server can simultaneously satisfy multiple users' demands. In fact, under the natural restriction of uncoded placement, the rate of the MN scheme is shown to be optimal when $K\leq N$ in \cite{WTP2016,WTP2020}. When a file is requested multiple times, Yu et al. \cite{YMA2018} designed a scheme by removing the redundant transmissions in the MN scheme, which is optimal under the constraint of uncoded placement and $K>N$.

However, the subpacketization of the MN scheme increases exponentially with the number of users $K$. To address the large subpacketization problem, \cite{YCTC} characterized a coded caching scheme with uncoded placement and one-shot delivery by a combinatorial array called placement delivery array (PDA). Subsequently, various coded caching schemes with lower subpacketization were proposed based the concept of PDA, such as \cite{SZG,YTCC,CJWY,CJYT,ASK,CJTY,CWZW,WCWC}.
Several variants of the centralized coded caching problem have been studied in the literature, including decentralized coded caching \cite{MN_D}, device to device coded caching \cite{JCM,WCYT}, online caching \cite{PMN}, caching with non-uniform file popularity and demands \cite{NM,ZLW}, multi-access coded caching \cite{HKD,RK,CWLZC,RK2021,SPE,SR_arX2021,SR,MR}, etc.

{\it Notations:}
\begin{itemize}
\item $|\cdot|$ denotes the cardinality of a set; $\min\{a,b\}$ denotes the smallest number in $\{a,b\}$.
\item If $a$ is not divisible by $q$, $\langle a\rangle _q$ denotes the least non-negative residue of $a$ modulo $q$; otherwise, $\langle a\rangle _q:=q$; $[a:b]:=\left\{ a,a+1,\ldots,b\right\}$; $[a:b]_q:=\{\langle a\rangle_q, \langle a+1\rangle_q, \ldots, \langle b\rangle_q\}$.
\item $\lfloor a\rfloor$ denotes the largest integer not greater than $a$; $\lceil a\rceil$ denotes the smallest integer not less than $a$; $a|b$ denotes $a$ divides $b$; $a\nmid b$ denotes $a$ does not divide $b$.
\item $\mathbf{P}(i,j)$ represents the element located in the $i^{\text{th}}$ row and the $j^{\text{th}}$ column of some array $\mathbf{P}$.
\end{itemize}

\subsection{Multi-access caching system}
The multi-access caching system was proposed in \cite{HKD}, which was motivated by the upcoming heterogeneous cellular architecture, which will include dense deployment of wireless access points (APs) and sparse cellular base stations (BSs). APs have small coverage and relatively large data rate, while BSs have large coverage and smaller data rate. The BS transmission rate can be reduced by placing caches at local APs with each user capable of accessing the content stored at multiple nearby APs in addition to receiving the BS broadcast. Compared to the centralized caching system, which is heavily limited by the memory size of each user's cache since a library size is usually much larger than the storage size of a mobile device, each user in the multi-access caching system has access to multiple ($L$) consecutive caches with a cyclic wrap-around, as shown in Fig. \ref{multiaccess-system}.
A central server with $N$ files (denoted by $W_1,W_2,\ldots,W_N$), each of size $1$ unit, is connected to $K$ cache-less users (denoted by $U_1,U_2,\ldots,U_K$) through an error-free shared-link. There are $K$ cache-nodes (denoted by $C_1,C_2,\ldots,C_{K}$), each has a memory size of $M$ units, and each user has access to $L$ consecutive cache-nodes with a cyclic wrap-around, i.e.,  each user $U_i$ has access to the cache-nodes $C_i,C_{\langle i+1\rangle_K},\ldots, C_{\langle i+L-1\rangle_K}$. Such a system is referred to as the $(K,L,M,N)$ multi-access caching system.

\begin{figure}
\centering
\includegraphics[width=4in]{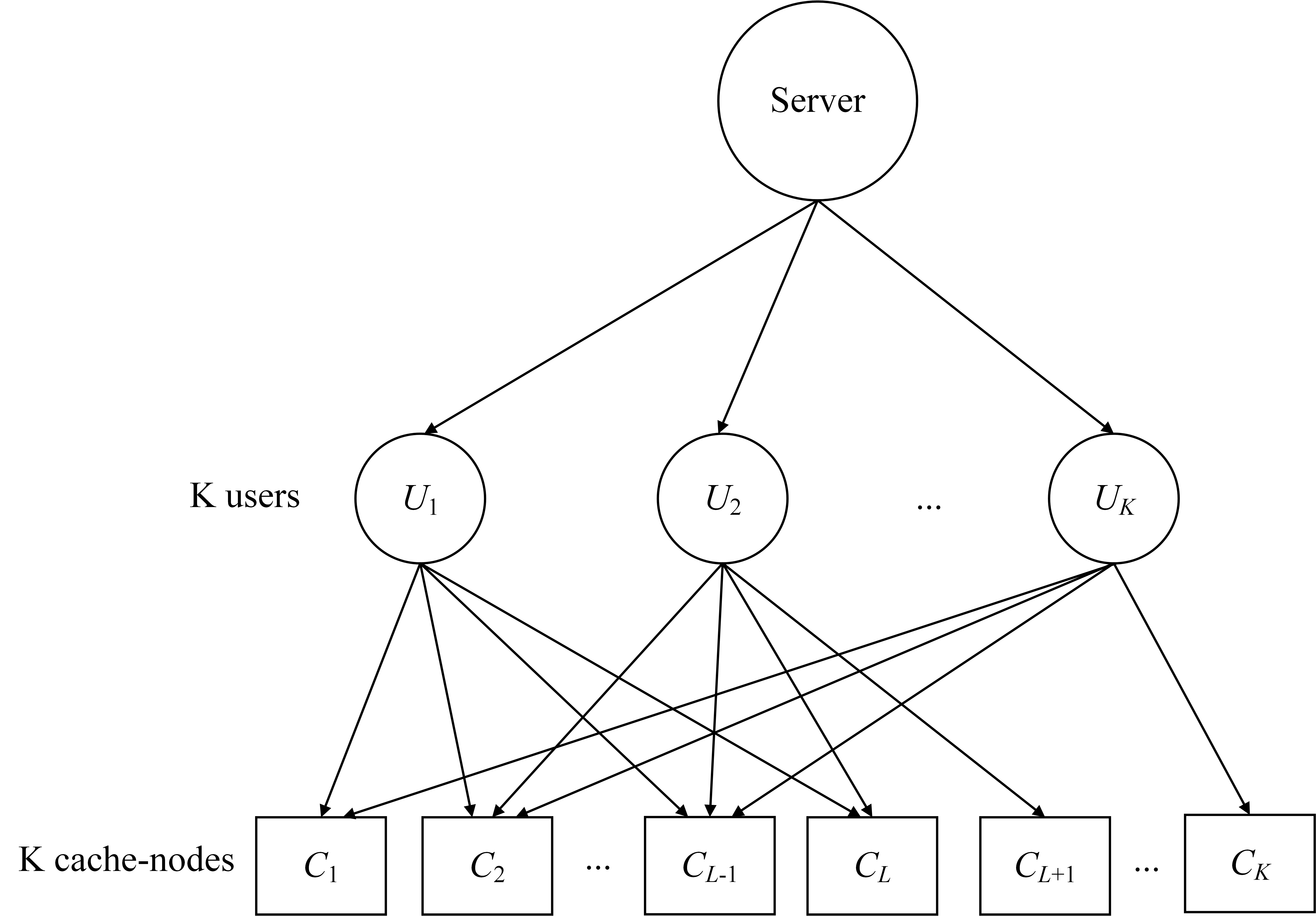}
\vskip 0.2cm
\caption{Multi-access caching system.}\label{multiaccess-system}
\vspace{-0.8cm}
\end{figure}

A $(K,L,M,N)$ multi-access coded caching scheme operates in two phases:
\begin{itemize}
\item {\bf Placement phase:} Each file is divided into $F$ packets of equal size, and some packets are stored in each cache-node without coding, which is referred to as uncoded placement. $\mathcal{Z}_{C_k}$ denotes the cache contents in cache-node $C_k$. The total size of different contents which can be retrieved by each user from its connected cache nodes is defined as the {\em local caching gain}, which is preferable to be as large as possible. The placement phase is done without knowledge of the users' later requests.
\item {\bf Delivery phase:} Each user requests one file from the server. Assume that user $U_k$ requests the file $W_{d_k}$, where $d_k\in[1:N]$, then the vector ${\bf d}=(d_1,d_2,\ldots,d_{K})$ is referred to as the {\em request vector}. According to the request vector ${\bf d}$ and the contents stored in the cache-nodes, the server broadcasts coded messages of total size $R_{\bf d}$ units to all users, such that each user can recover its requested file. The worst case transmission amount $R=\max_{{\bf d} \in [1:N]^K} R_{{\bf d}}$ is referred to as the {\em transmission rate}. The average number of users served by each message is defined as the {\em coded caching gain}, which is also preferable to be as large as possible.
\end{itemize}

\subsection{Previous Results}
The first multi-access caching scheme was proposed in \cite{HKD}, which achieves a transmission rate equal to
\begin{equation}
\label{HKD}
R_{HKD}=\begin{cases}
\frac{K(1-\frac{LM}{N})}{1+\frac{KM}{N}}, \ \text{if} \ L|K\\
\frac{K(1-\frac{M}{N})}{1+\frac{KM}{N}}, \ \text{otherwise}
\end{cases}
\end{equation}
for any $M=\frac{\gamma}{K}N$ where $\gamma\in\left[0:\left\lfloor\frac{K}{L}\right\rfloor\right]$, while the required subpacketization is $F_{HKD}=O({K\choose KM/N})$.
By means of index coding, \cite{RK} proposed a scheme achieving the following transmission rate
\begin{equation}
\label{RK}
R_{RK1}=K\left(1-\frac{LM}{N}\right)^2, \ \forall M=\frac{\gamma}{K}N,\ \gamma\in\left[0:\left\lfloor\frac{K}{L}\right\rfloor\right],
\end{equation}
 which is smaller than the rate in \eqref{HKD} when $K<\frac{KLM}{N}+L$, while the required subpacketization is
 \begin{equation}
 \label{F_RK}
 F_{RK1}=\frac{N}{M}{K-KM/N(L-1)-1\choose KM/N-1}.
 \end{equation}
 Moreover, when $L\geq \frac{K}{2}$ and $N \geq K$, a converse bound for the multi-access caching system under uncoded placement was derived in \cite{RK} and the proposed scheme was shown to be order optimal within a factor of $2$.
 In \cite{CWLZC}, the authors proposed a transformation approach to extend the MN scheme to a multi-access caching scheme with the rate
 \begin{equation}
 \label{CW}
 R_{CW}=\frac{K(1-\frac{LM}{N})}{1+\frac{KM}{N}},\ \forall M=\frac{\gamma}{K}N,\ \gamma\in\left[0:\left\lfloor\frac{K}{L}\right\rfloor\right],
 \end{equation}
 which is smaller than the rate in \eqref{HKD} (when $L\nmid K$ and $L>1$) and the rate in \eqref{RK} (when $K>\frac{KLM}{N}+L$), the required subpacketization is $F_{CW}=K{K-(L-1)KM/N\choose KM/N}$. By dividing the multi-access coded caching problem into a number of special class of index coding problems termed as Structured Index Coding problem and using the solutions obtained for the structured index coding problem, the authors in \cite{RK2021} proposed a new scheme with the rate not greater than the rates in \eqref{RK} and \eqref{CW}, and the required subpacketization is the same as $F_{RK1}$ in \eqref{F_RK}.

In addition, for some special parameters, some multi-access caching schemes with lower subpacketization were proposed. For example, \cite{SPE} studied the case when $\frac{M}{N}=\frac{2}{K}$ (i.e., $\gamma=\frac{KM}{N}=2$), and proposed a scheme achieving the rate $R_{SPE}=\frac{K(1-LM/N)}{g_{SPE}}$ (where the coded caching gain $g_{SPE}>\gamma+1=3$), which is strictly lower than the rate in \eqref{CW},  the required subpacketization is $F_{SPE}=\frac{K(K-2L+2)}{4}$;
\cite{SR_arX2021} studied the case when $\frac{M}{N}=\frac{\gamma}{K}$ where $\gamma$ and $K$ are coprime, and proposed a scheme with the rate $R_{SR1}$ not greater than that in \eqref{RK}, the required subpacketization satisfies $F_{SR1}\leq K^2$; \cite{SR} studied the case when $\frac{M}{N}=\frac{\gamma}{K}$ where $\gamma|K$ and $(K-\gamma L+\gamma)|K$, and proposed a scheme with the rate $R_{SR2}=\frac{K(1-LM/N)(1-(L-1)M/N)}{2}$ and subpacketization $F_{SR2}=K$;
\cite{MR} considered the case when $\frac{M}{N}=\frac{1}{K}$, and proposed a scheme with the rate $R_{MR}=\lceil\frac{K(K-L)}{2+\lfloor\frac{L}{K-L+1}\rfloor+\lfloor\frac{L-1}{K-L+1}\rfloor}\rceil\frac{1}{K}$ and subpacketization $F_{MR}=K$. In fact, the MR scheme in \cite{MR} can be applied to general parameters by carefully designing the placement.










\subsection{Contributions}
In this paper, we consider the the $(K,L,M,N)$ multi-access caching system and our contributions are summarized below.
\begin{itemize}
\item In order to achieve a linear subpacketization and the maximal local caching gain (i.e., the cached contents at any $L$ neighbouring cache-nodes are different such that each user can totally retrieve $LM$ files from its connected cache-nodes), we propose a placement strategy called the {\em consecutive cyclic placement} (i.e., Definition \ref{conti_cyc} in Section \ref{main_r}), which can be represented by a cache-node placement array and a user-retrieve array (similar to PDA).
\item Based on the user-retrieve array under the consecutive cyclic placement, we derive the optimal (maximal) coded caching gain from the perspective of PDA, thus obtaining a tight lower bound on the rate of PDA (i.e., Theorem \ref{maxgain} in Section \ref{main_r}).
\item According to the derived lower bound, we construct a class of PDA under the consecutive cyclic placement, leading to a multi-access coded caching scheme (i.e., Theorem \ref{ach_R_F} in Section \ref{main_r}), which achieves our derived lower bound for some parameters. For other parameters, the achieved coded caching gain is only 1 less than the optimal one. Moreover, the needed subpacketization is less than $K^2$.
\item We provide analytical and numerical comparisons of the proposed scheme with the state-of-the-art. Compared to the schemes with exponential subpacketization, we show that our scheme has a lower rate and a lower subpacketization than the schemes in \cite{HKD,RK}, and has a slightly higher rate and a lower subpacketization than the schemes in \cite{CWLZC,RK2021}. Compared to the schemes with linear subpacketization (less than $K^2$), we show that our scheme has a larger coded caching gain than the schemes in \cite{SR,SR_arX2021,MR}, and has a smaller coded caching gain than the SPE scheme in \cite{SPE}. However, the SPE scheme is only applicable to specific memory ratio, i.e., $\frac{M}{N}=\frac{2}{K}$, while our scheme is applicable to arbitrary memory ratio in $\{\frac{\gamma}{K}|\gamma\in[0:\lfloor\frac{K}{L}\rfloor]\}$.
\end{itemize}

The rest of this paper is organized as follows. In Section \ref{intr_PDA}, we introduce the definition of PDA and the relationship between a PDA and a centralized coded caching scheme. In Section \ref{main_r}, we list the main results of this paper and provide the performance analysis of the proposed scheme in Theorem \ref{ach_R_F}. In Section \ref{de_scheme}, we give the construction of PDA under the consecutive cyclic placement in two cases and provide the proof of Theorem \ref{ach_R_F}. Finally, Section \ref{conclusion} concludes the paper and some proofs are provided in the Appendices.

\section{Placement and delivery array}
\label{intr_PDA}
Placement delivery array (PDA) was originally proposed in \cite{YCTC}, which is a combinatorial array used to design coded caching schemes with uncoded placement and one-shot delivery, i.e., each user can recover any requested file packet from at most one transmitted message with the help of the contents it has access to.
\begin{definition}(\cite{YCTC})
\label{def-PDA}
For  positive integers $K,F, Z$ and $S$, an $F\times K$ array  $\mathbf{P}$ composed of a specific symbol $``*"$  and $S$ integers in $[1:S]$, is called a $(K,F,Z,S)$ placement delivery array (PDA) if it satisfies the following conditions:
\begin{enumerate}
  \item [{\bf C1:}] The symbol $``*"$ appears $Z$ times in each column;
  \item [{\bf C2:}] Each integer in $[1:S]$ occurs at least once in the array;
  \item [{\bf C3:}] For any two distinct entries $\mathbf{P}(j_1,k_1)$ and $\mathbf{P}(j_2,k_2)$, if    $\mathbf{P}(j_1,k_1)=\mathbf{P}(j_2,k_2)=s\in[1:S]$, then $\mathbf{P}(j_1,k_2)=\mathbf{P}(j_2,k_1)=*$, i.e., the corresponding $2\times 2$  subarray formed by rows $j_1,j_2$ and columns $k_1,k_2$ must be of the following form
    \begin{eqnarray*}
    \left(\begin{array}{cc}
      s & *\\
      * & s
    \end{array}\right)~\textrm{or}~
    \left(\begin{array}{cc}
      * & s\\
      s & *
    \end{array}\right).
  \end{eqnarray*}

\end{enumerate}
Furthermore, if each integer appears exactly $g$ times in $\mathbf{P}$, $\mathbf{P}$ is called a $g$-regular $(K,F,Z,S)$ PDA, $g$-$(K,F,Z,S)$ PDA or $g$-PDA for short.
\end{definition}

A $(K,F,Z,S)$ PDA $\mathbf{P}$ can generate a coded caching scheme as follows.
\begin{itemize}
\item {\bf Placement phase:} Split each file $W_n$ into $F$ packets, i.e., $W_{n}=\{W_{n,j}\ |\ j\in[1:F]\}$. User $U_k$ has access to $\mathcal{Z}_{U_k}=\{W_{n,j}\ |\ \mathbf{P}(j,k)=*, j\in[1:F], n\in [1:N]\}$, where $k\in [1:K]$.
\item {\bf Delivery phase:} Assume that the request vector is ${\bf d}=(d_1,d_2,\ldots,d_K)$, for each $s\in[1:S]$, the server broadcasts $\bigoplus_{\mathbf{P}(j,k)=s,j\in[1:F],k\in[1:K]}W_{d_{k},j}$ to all users at time slot $s$.
\end{itemize}

\begin{lemma}(\cite{YCTC})
\label{th-Fundamental}
A $(K,F,Z,S)$ PDA can generate a coded caching scheme with the user accessible memory ratio $\frac{Z}{F}$, subpacketization $F$ and  rate $R=\frac{S}{F}$.
\end{lemma}

For a $(K,F,Z,S)$ PDA $\mathbf{P}$, columns and rows represent the user indexes and the packet indexes respectively. If $\mathbf{P}(j,k)=*$, it means that user $U_k$ has access to the $j^{\text{th}}$ packet of all files. Condition C1 of Definition \ref{def-PDA} implies that each user has access to the same memory size and the user accessible memory ratio is $\frac{Z}{F}$. If $\mathbf{P}(j,k)=s$ is an integer, it implies that the $j^{\text{th}}$ packet of all files is not accessible to user $U_k$, and the server broadcasts a multicast message (i.e. the XOR of all the requested packets indicated by $s$) to all the users at time slot $s$. Condition C3 of Definition \ref{def-PDA} ensures that each user can get its desired packet, since all the other packets in the multicast message are accessible to it. The occurrence number of integer $s$ in $\mathbf{P}$, denoted by $g_s$, is the coded caching gain at time slot $s$, since the message broadcasted at time slot $s$ can serve $g_s$ users simultaneously. Condition C2 of Definition \ref{def-PDA} implies that the number of messages broadcasted by the server is exactly $S$, so the rate is $R=\frac{S}{F}$.

\section{Main results}
\label{main_r}
In this paper, we focus on the $(K,L,M,N)$ multi-access caching system with cache-node memory size $M=\frac{\gamma N}{K}$ for $\gamma\in [0:\lfloor \frac{K}{L}\rfloor]$. In order to achieve a linear subpacketization and the maximal local caching gain, we propose the following placement strategy: each file $W_n$ is split into $K$ equal-length subfiles, i.e., $W_n=\{W_{n,i}|i\in[1:K]\}$, and cache-node $C_k$ caches
\begin{equation}
\label{eq_cache}
\mathcal{Z}_{C_k}=\{W_{n,\langle k-L+1\rangle_K},W_{n,\langle k-2L+1\rangle_K},\ldots,W_{n,\langle k-\gamma L+1\rangle_K}|n\in[1:N]\},
\end{equation}
for any $k\in[1:K]$. Then the total size of the contents stored by each cache-node is $\frac{\gamma N}{K}=M$ units, which satisfies the memory size constraint. Furthermore, any $L$ neighbouring cache-nodes do not cache any same subfile, then each user can retrieve $LM$ units from its connected cache-nodes, thus achieving the maximal local caching gain.

 \begin{definition}
 \label{conti_cyc}
 For a $(K,L,M,N)$ multi-access caching system, for any cache-node memory size $M=\frac{\gamma}{K}N$ where $\gamma\in[0:\lfloor\frac{K}{L}\rfloor]$, if each file is split into $K$ equal-length subfiles, and each cache-node $C_k$ caches $\mathcal{Z}_{C_k}$ in \eqref{eq_cache}, i.e., the $j^{\text{th}}$ subfile of all files is stored by the cache-nodes $C_{\langle j+L-1\rangle_K}, C_{\langle j+2L-1\rangle_K}, \ldots, C_{\langle j+\gamma L-1\rangle_K}$, so it can be retrieved by the users $U_j,U_{\langle j+1\rangle_K},\ldots,$ $U_{\langle j+\gamma L-1\rangle_K}$, the placement is called the {\em consecutive cyclic placement}.
 \end{definition}
Similar to PDA, the consecutive cyclic placement can be represented by a cache-node placement array and a user-retrieve array. For example, when $K=10, L=3, \gamma=2$, the cache-node placement array $\mathbf{C}$ and the user-retrieve array $\mathbf{U}$ are shown in Fig \ref{ca_node}, where $\mathbf{C}(j,k)=*$ represents that the $j^{\text{th}}$ subfile of all files is cached by cache-node $C_k$, and $U(j,k)=*$ represents that the $j^{\text{th}}$ subfile of all files can be retrieved by user $U_k$.
\begin{figure}
  \centering
  \includegraphics[width=5.5in]{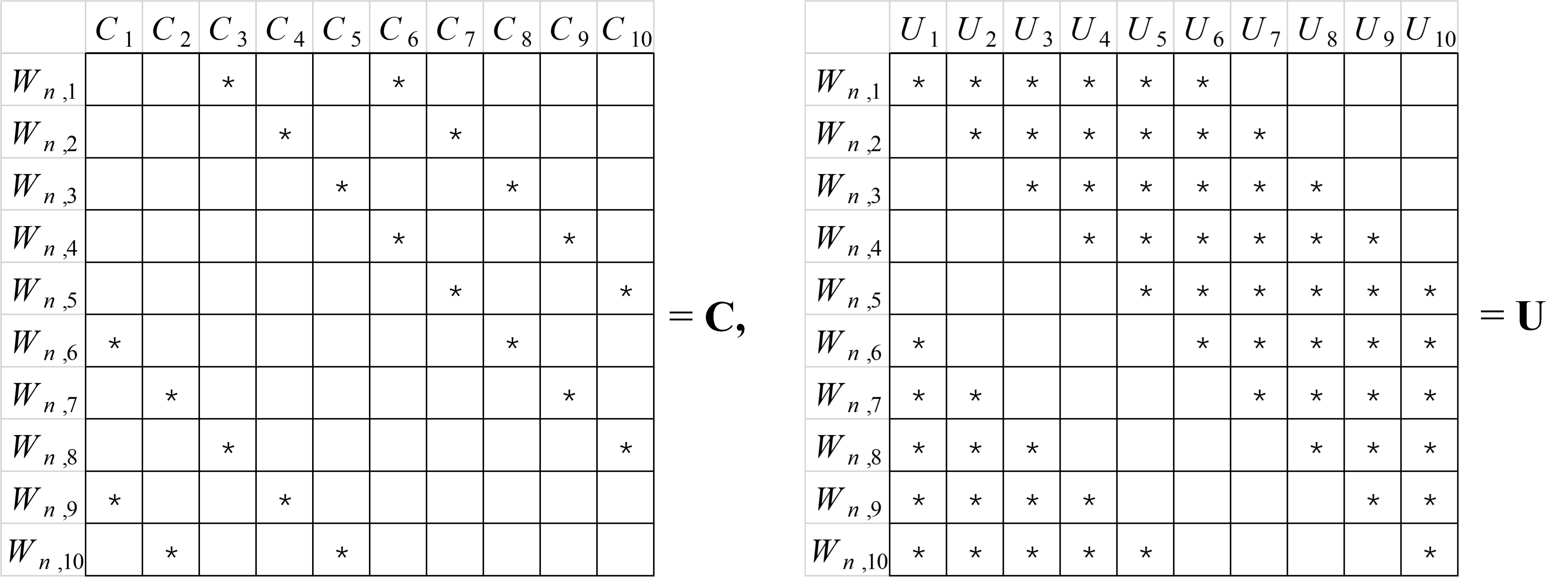}
  \caption{The cache-node placement array $\mathbf{C}$ and the user-retrieve array $\mathbf{U}$ when $K=10,L=3,\gamma=2$.}\label{ca_node}
\vspace{-0.8cm}
\end{figure}

Under the consecutive cyclic placement, we derive the maximal coded caching gain from the perspective of PDA, and obtain a lower bound on the rate of PDA .
\begin{theorem}
\label{maxgain}
For a $(K,L,M,N)$ multi-access caching system, for any cache-node memory size $M=\frac{\gamma}{K}N$ where $\gamma\in[0:\lfloor\frac{K}{L}\rfloor]$, let $t=\gamma L$, the maximal coded caching gain $g$ of any PDA under the consecutive cyclic placement satisfies
\begin{equation}
\label{gupperbound}
g\leq g^*=\begin{cases}
2\lfloor\frac{K}{K-t+1}\rfloor, \ \ \text{if} \ \langle K\rangle_{K-t+1}\leq \lfloor\frac{K-t}{2}\rfloor \ \text{or} \ (K-t+1)|K \\
2\lfloor\frac{K}{K-t+1}\rfloor+1, \ \ \text{otherwise}, \\
\end{cases}
\end{equation}
thus, the achievable rate $R$ of any PDA under the consecutive cyclic placement satisfies
\begin{equation}
\label{Rlowerbound}
R\geq R^*=\begin{cases}
\frac{K-t}{2\lfloor\frac{K}{K-t+1}\rfloor}, \ \ \text{if} \ \langle K\rangle_{K-t+1}\leq \lfloor\frac{K-t}{2}\rfloor \ \text{or} \ (K-t+1)|K\\
\frac{K-t}{2\lfloor\frac{K}{K-t+1}\rfloor+1}, \ \ \text{otherwise}. \\
\end{cases}
\end{equation}
\end{theorem}
The proof of Theorem \ref{maxgain} is given in Appendix \ref{pr_Theorem_Rlowerbound}. According to the lower bound in Theorem \ref{maxgain}, we construct a class of PDA under the consecutive cyclic placement, which generates the following multi-access caching scheme. It is worth noting that for some cases we need to further divide each subfile into several packets, so that we can construct a regular PDA (i.e., each non-star entry appears the same times).

\begin{theorem}
\label{ach_R_F}
For a $(K,L,M,N)$ multi-access caching system, for any cache-node memory size $M=\frac{\gamma}{K}N$ where $\gamma\in[0:\lfloor\frac{K}{L}\rfloor]$, let $t=\gamma L$, the rate in \eqref{ach_rate} is achievable under the consecutive cyclic placement with the subpacketization in \eqref{packet_number}.
\begin{figure}
\begin{equation}
\label{ach_rate}
R_{new}=\begin{cases}
\frac{(K-t)(K-t+1)}{2K}, \ \ \text{if} \ (K-t+1)|K \ \text{or} \ K-t=1\\
\frac{K-t}{2\lfloor\frac{K}{K-t+1}\rfloor+1}, \ \ \text{if} \ \langle K\rangle_{K-t+1}=K-t \ \text{and} \ K-t>1\\
\frac{K-t}{2\lfloor\frac{K}{K-t+1}\rfloor}, \ \ \text{otherwise.}
\end{cases}
\end{equation}
\vspace{-1cm}
\end{figure}
\begin{figure}
\begin{equation}
\label{packet_number}
F_{new}=\begin{cases}
K,\ \ \ \ \text{if} \ (K-t+1)|K \ \text{or} \ K-t=1\\
\left(2\lfloor\frac{K}{K-t+1}\rfloor+1\right)K, \ \ \text{if} \ \langle K\rangle_{K-t+1}=K-t\ \text{and} \ K-t>1\\
2\lfloor\frac{K}{K-t+1}\rfloor K, \ \ \text{otherwise}.
\end{cases}
\end{equation}
\vspace{-1cm}
\end{figure}
\end{theorem}

The proof of Theorem \ref{ach_R_F} is given in Section \ref{pr_th_ach_R_F}.
\begin{remark}
It is worth noting that when $\langle K\rangle_{K-t+1}\leq \lfloor\frac{K-t}{2}\rfloor$ or $\langle K\rangle_{K-t+1}=K-t$ or $(K-t+1)|K$ or $K-t=1$, the rate of the proposed scheme in Theorem \ref{ach_R_F} achieves the lower bound in Theorem \ref{maxgain}. For other parameters, the achieved coded caching gain of the proposed scheme in Theorem \ref{ach_R_F} is $g_{new}=2\lfloor\frac{K}{K-t+1}\rfloor$, which is only $1$ less than the optimal coded caching gain (i.e., $g^*=2\lfloor\frac{K}{K-t+1}\rfloor+1$) in Theorem \ref{maxgain}. Moreover, the subpacketization of the proposed scheme is less than $K^2$.
\end{remark}

Next we will compare our scheme with the state-of-the-art. Since the rate expression of the schemes in \cite{RK2021,MR,SR_arX2021,SPE} is complicated, we are not able to compare our scheme with these schemes analytically.
\begin{itemize}
\item {\it Comparison to the RK1 scheme in \cite{RK}:}
When $\frac{M}{N}=\frac{\gamma}{K}$, the rate of the RK1 scheme (given in \eqref{RK}) is $R_{RK1}=(K-\gamma L)(1-\frac{\gamma L}{K})$, while the rate of the proposed scheme in Theorem \ref{ach_R_F} is $R_{new}\leq \frac{K-\gamma L}{2\lfloor\frac{K}{K-\gamma L+1}\rfloor}$. Hence, when $2\lfloor\frac{K}{K-\gamma L+1}\rfloor>\frac{K}{K-\gamma L}$, we have $R_{new}<R_{RK1}$. In particular, when $(K-\gamma L+1)|K$ and $K-\gamma L>1$, we have $R_{new}<R_{RK1}$. Moreover, in this case the subpacketization of the proposed scheme is $K$, while the subpacketization of the RK1 scheme is  $F_{RK1}=\frac{K}{\gamma}{K-\gamma(L-1)-1\choose \gamma-1}$, which is exponential with the number of users.

\item {\it Comparison to the CW scheme in \cite{CWLZC}:}
When $\frac{M}{N}=\frac{\gamma}{K}$, the rate of the CW scheme in \cite{CWLZC} (given in \eqref{CW}) is $R_{CW}=\frac{K-\gamma L}{\gamma+1}$, while the rate of the proposed scheme is $R_{new}\leq \frac{K-\gamma L}{2\lfloor\frac{K}{K-\gamma L+1}\rfloor}$.
\begin{itemize}
\item When $\gamma=1$, we have $2\lfloor\frac{K}{K-\gamma L+1}\rfloor\geq 2=\gamma+1$, then $R_{new}\leq R_{CW}$.
\item When $\gamma>1$ and $\gamma$ is odd, if $K\leq \gamma L+2L-1+\frac{2(L-1)}{\gamma-1}$ (i.e., $\frac{K}{K-\gamma L+1}\geq \frac{\gamma+1}{2}$, which implies $2\lfloor\frac{K}{K-\gamma L+1}\rfloor\geq \gamma+1$), then $R_{new}\leq R_{CW}$.
\item When $\gamma>1$ and $\gamma$ is even, if $K\leq \gamma L+2L-1-\frac{2}{\gamma}$ (i.e., $\frac{K}{K-\gamma L+1}\geq \frac{\gamma+2}{2}$, which implies $2\lfloor\frac{K}{K-\gamma L+1}\rfloor> \gamma+1$), then $R_{new}<R_{CW}$.
\end{itemize}
Moreover, the subpacketization level of the proposed scheme is less than $K^2$, while the subpacketization level of the CW scheme is $F_{CW}=K{K-\gamma(L-1)\choose \gamma}$, which is exponential with the number of users.

\item {\it Comparison to the SR2 scheme in \cite{SR}:}
When $\gamma|K$ and $(K-\gamma L+\gamma)|K$, the SR2 scheme in \cite{SR} has the rate $R_{SR2}=\frac{(K-\gamma L)(K-\gamma L+\gamma)}{2K}$ and subpacketization $R_{SR2}=K$.
\begin{itemize}
\item If $\gamma=1$, we have $R_{new}=R_{SR2}$ and $F_{new}=F_{SR2}$.
\item If $\gamma>1$ and $(K-\gamma L+1)|K$, we have $R_{new}=\frac{(K-\gamma L)(K-\gamma L+1)}{2K}$ and $F_{new}=K$, then $R_{new}<R_{SR2}$ and $F_{new}=F_{SR2}$.
\item If $\gamma>1$ and $\langle K\rangle_{K-\gamma L+1}=K-\gamma L$, we have $R_{new}=\frac{K-\gamma L}{2\lfloor\frac{K}{K-\gamma L+1}\rfloor+1}\leq \frac{K-\gamma L}{\frac{2 K}{K-\gamma L+\gamma}+1}$ and $F_{new}=(2\lfloor\frac{K}{K-\gamma L+1}\rfloor+1)K$, then $R_{new}<R_{SR2}$ and $F_{new}>F_{SR2}$.
\item For other cases, we have $R_{new}=\frac{K-\gamma L}{2\lfloor\frac{K}{K-\gamma L+1}\rfloor}\leq \frac{K-\gamma L}{\frac{2 K}{K-\gamma L+\gamma}}$ and $F_{new}=2\lfloor\frac{K}{K-\gamma L+1}\rfloor K$, then $R_{new}\leq R_{SR2}$ and $F_{new}>F_{SR2}$.
\end{itemize}
Moreover, the SR2 scheme is only applicable to specific memory ratios, while the proposed scheme is applicable to arbitrary memory ratio in $\{\frac{\gamma}{K}|\gamma\in[0:\lfloor\frac{K}{L}\rfloor]\}$.

\item {\it Numerical comparison with the schemes in \cite{HKD,RK,CWLZC,RK2021}:}
When $K=36$ and $L=5$, the memory-rate and memory-subpacketization tradeoffs of the proposed scheme and the schemes in \cite{HKD,RK,CWLZC,RK2021} are shown in Fig \ref{compareRF}. It can be seen that our scheme has a lower rate and simultaneously a lower subpacketization than the schemes in \cite{HKD,RK}, and has a lower subpacketization and a slightly higher rate than the schemes in \cite{CWLZC,RK2021}.

\begin{figure}
\centering
\includegraphics[width=3in]{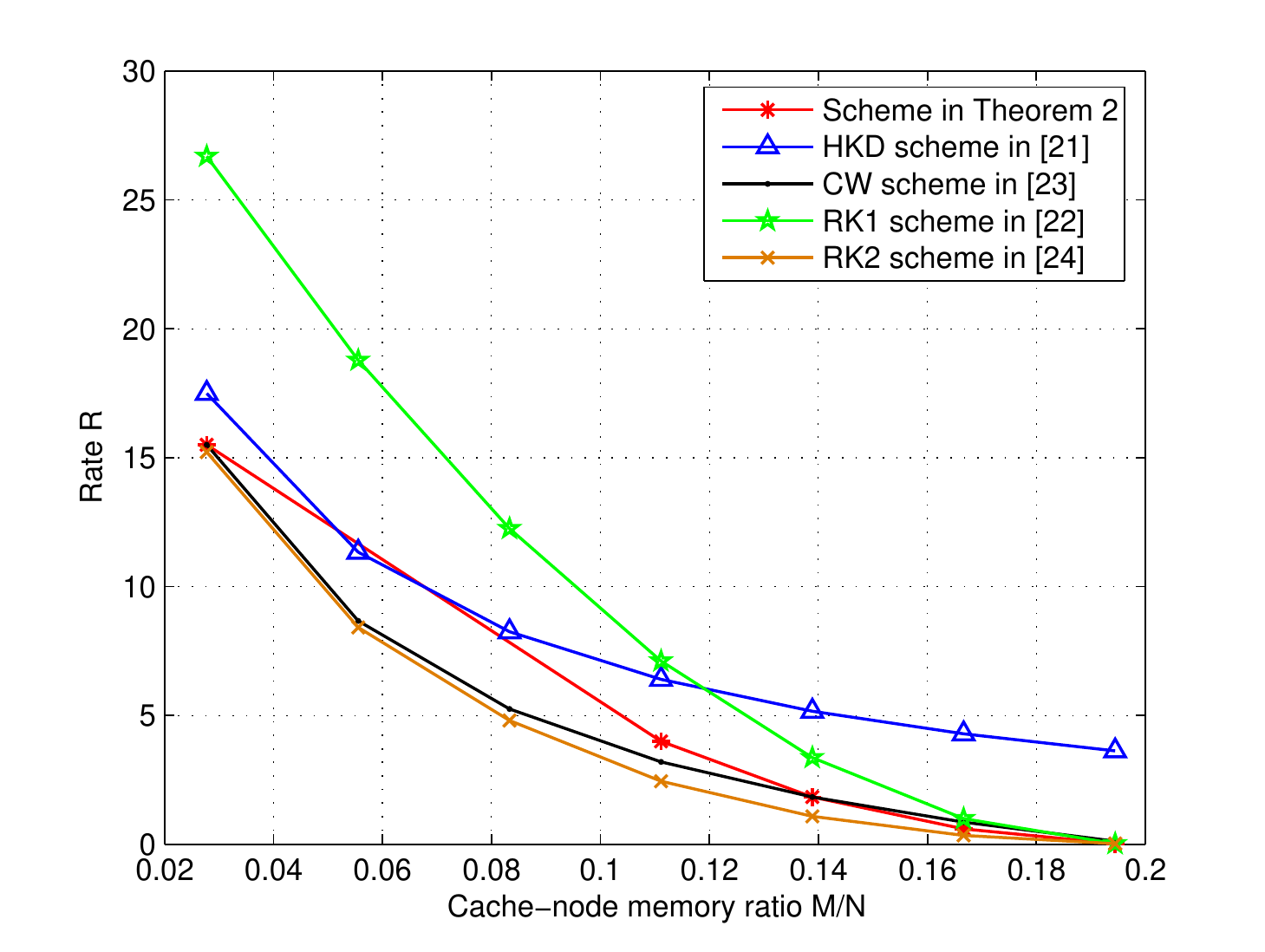}
\hspace{0.1in}
\includegraphics[width=3in]{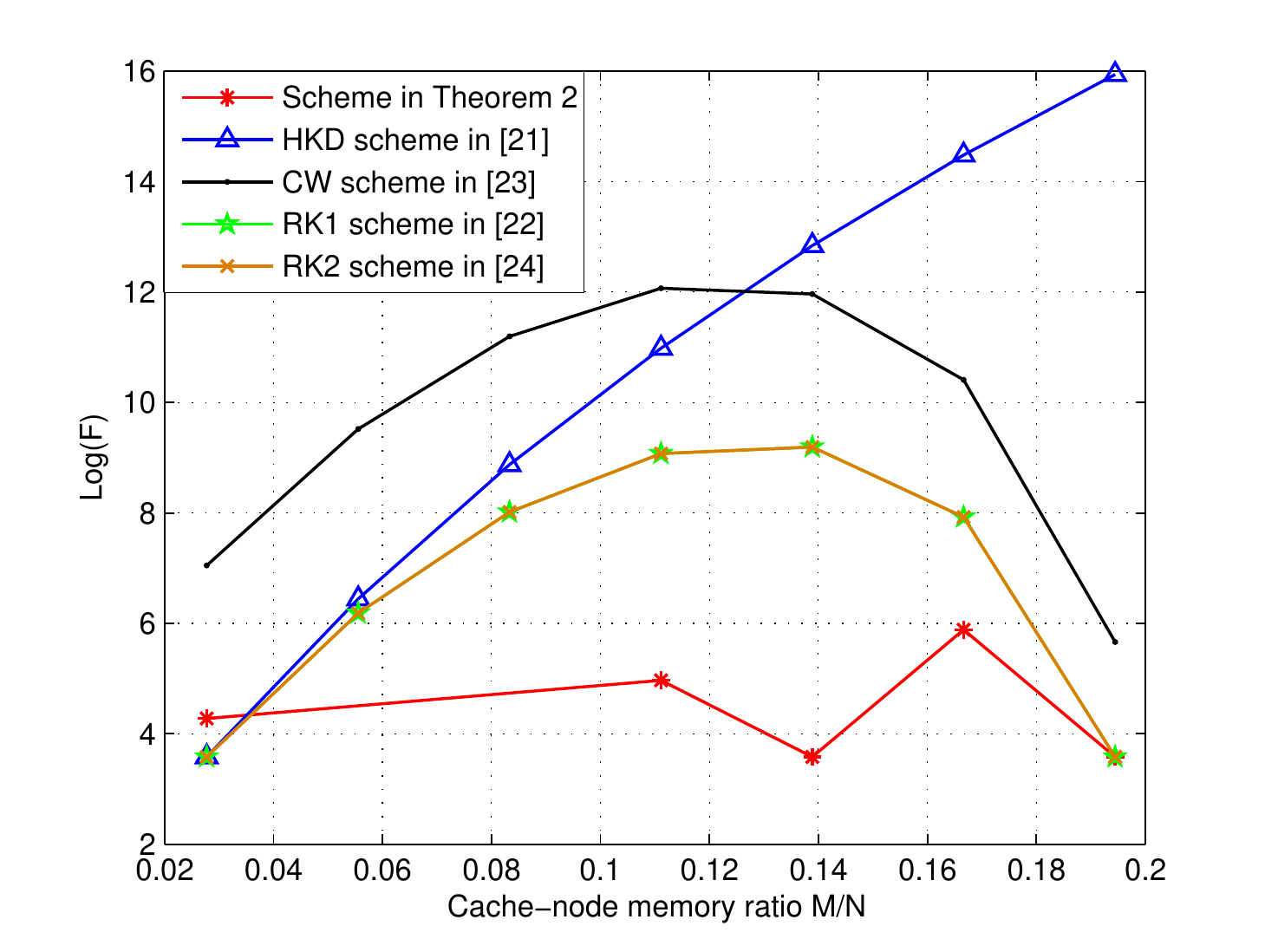}
\caption{The memory-rate and memory-subpacketization tradeoffs when $K=36$ and $L=5$.}\label{compareRF}
\vspace{-0.8cm}
\end{figure}

\item {\it Numerical comparison with the schemes in \cite{MR,SPE,SR,SR_arX2021}:}
Since the subpacketizations of the proposed scheme and the schemes in \cite{MR,SPE,SR,SR_arX2021} are all less than $K^2$, we only compare their coded caching gains. When $K=45$ and $L=7$, the memory ratio versus the coded caching gain of these schemes are shown in Fig \ref{compareg}. It can be seen that our scheme has a larger coded caching gain than the schemes in \cite{MR,SR,SR_arX2021}, and has a lower coded caching gain than the SPE scheme in \cite{SPE}. However, the SPE scheme is only applicable to specific memory ratio, i.e., $\frac{M}{N}=\frac{2}{K}$, while our scheme is applicable to arbitrary memory ratio in $\{\frac{\gamma}{K}|\gamma\in[0:\lfloor\frac{K}{L}\rfloor]\}$.
\begin{figure}
\centering
\includegraphics[width=3in]{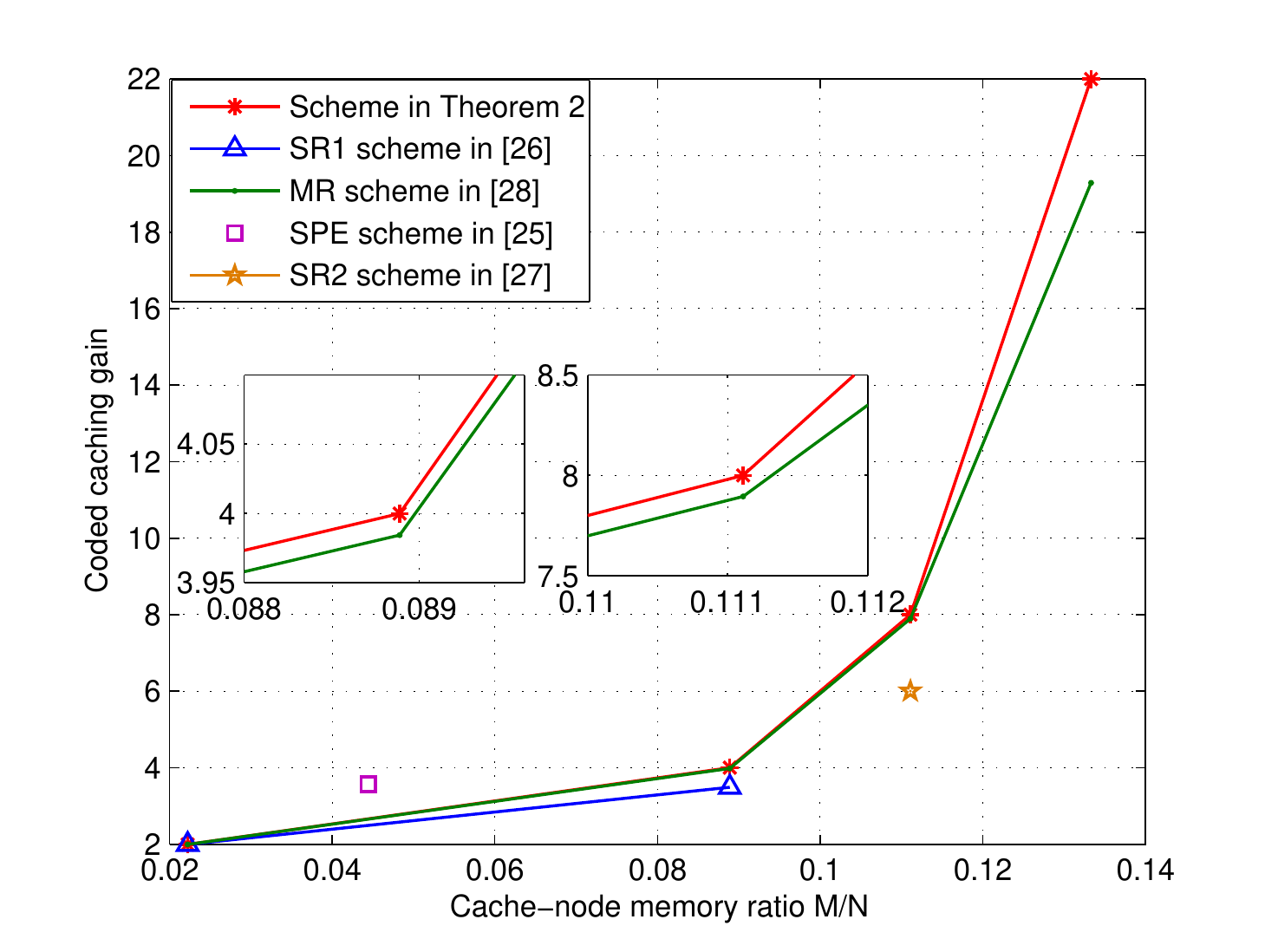}
\vskip 0.2cm
\caption{The cache-node memory ratio versus the coded caching gain when $K=45,L=7$.}\label{compareg}
\vspace{-0.8cm}
\end{figure}
\end{itemize}

\section{The proposed scheme in Theorem \ref{ach_R_F}}
\label{de_scheme}
For a $(K,L,M,N)$ multi-access caching system, for any cache-node memory size $M=\frac{\gamma}{K}N$ where $\gamma\in[0:\lfloor\frac{K}{L}\rfloor]$, let $t=\gamma L$, we will give the construction of PDA (which generates the proposed scheme in Theorem \ref{ach_R_F}) in two cases.
\subsection{The case of $(K-t+1)|K$ or $K-t=1$}
\begin{example}
\label{example2}
When $K=10, L=3, \gamma=2$, we have $t=\gamma L=6$ and $(K-t+1)|K$, the user-retrieve array $\mathbf{U}$ under the consecutive cyclic placement is shown in \eqref{ca_node}. In order to design the delivery strategy, we fill in two-dimensional vectors at the non-star positions in the user-retrieve array $\mathbf{U}$ such that the resulting array is a PDA (i.e., satisfies Condition C3 of Definition \ref{def-PDA}). Precisely, it can be done in three steps, as illustrated in Fig \ref{fill}.
\begin{figure}
  \centering
  \includegraphics[width=6in]{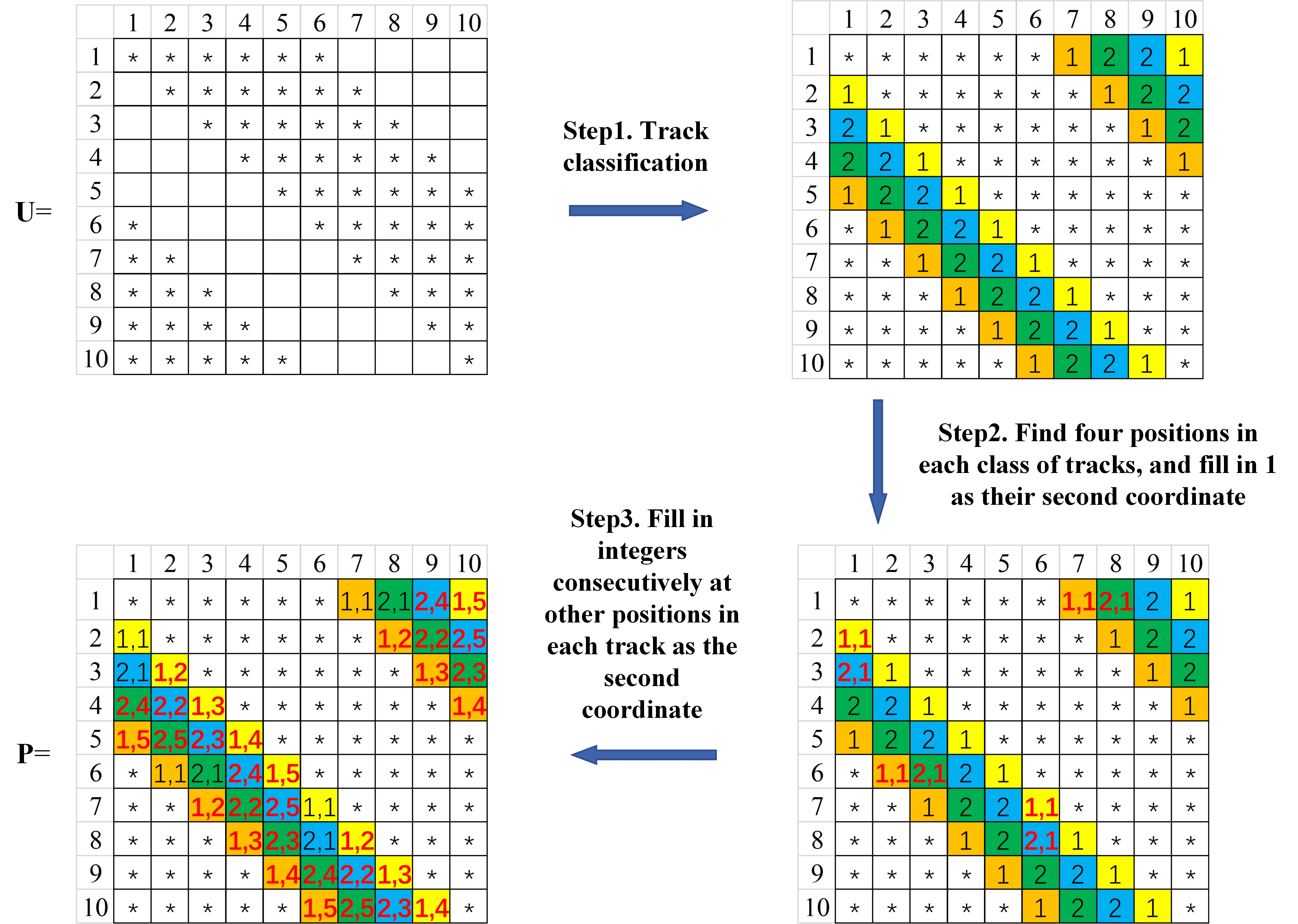}
  \caption{The steps of filling in two-dimensional vectors at the non-star positions in the user-retrieve array $\mathbf{U}$ when $K=10,L=3,\gamma=2$.}\label{fill}
\vspace{-0.8cm}
\end{figure}

{\bf Step1.} We classify the non-star positions in $\mathbf{U}$ into $K-t=4$ tracks. The $l^{\text{th}}$ track is the collection of positions $(j,k)$ satisfying $\langle j-k\rangle_K=l$ where $l\in[1:4]$; the first, second, third and fourth tracks are painted yellow, blue, green and orange respectively in Fig \ref{fill}. Then we classify the four tracks into two classes, i.e., the $l^{\text{th}}$ track belongs to the $\min\{l,K-t+1-l\}=\min\{l,5-l\}^{\text{th}}$ class. The category number of each track is written into each position belonging to that track as the first coordinate of the two-dimensional vector.

{\bf Step2.} For each class of tracks, we first find $\frac{2K}{K-t+1}=4$ positions (in different columns and different rows) in the order of the slowest increase of column labels, such that the element at the intersection of the row of any position and the column of another position is a star. Then we fill in $1$ at the found positions as the second coordinate. For example, for the first class of tracks, we find four positions in turn: $(2,1)\rightarrow (6,2) \rightarrow (7,6) \rightarrow (1,7)$, which form the subarray in \eqref{subarray1}, and we fill in $1$ at the non-star positions in \eqref{subarray1} as the second coordinate, then the vector $(1,1)$ appears four times in the resulting subarray and satisfies Condition C3 of Definition \ref{def-PDA}.
\begin{equation}
\label{subarray1}
\bordermatrix{%
  & 1 & 2 & 6 & 7 \cr
1 & * & * & * & 1 \cr
2 & 1 & * & * & * \cr
6 & * & 1 & * & * \cr
7 & * & * & 1 & *}.
\end{equation}

{\bf Step3.} We fill in integers consecutively at other positions (whose second coordinate has not been filled in yet) in each track as  the second coordinate. Precisely, if the vector at the position $(j,k)$ is $(l,1)$, then the vector at the position $(\langle j+\alpha\rangle_K, \langle k+\alpha\rangle_K)$ in the same track should be $(l,1+\alpha)$, where $\alpha\in[1:K-t]=[1:4]$. Since the vector $(l,1)$ satisfies Condition C3, the consecutive cyclic placement guarantees the vector $(l,1+\alpha)$ also satisfies Condition C3. The constraint $(K-t+1)|K$ guarantees that each vector appears exactly $\frac{2K}{K-t+1}=4$ times in the resulting array $\mathbf{P}$.

It is easy to verify that the resulting array $\mathbf{P}$ is a $4$-$(10,10,6,10)$ PDA under the consecutive cyclic placement, which leads to a multi-access coded caching scheme with the rate $R=1$ and subpacketization $F=10$.
\end{example}

In general, the mathematical representation of the above construction is as follows:
\begin{construction}
\label{constr1}
For a $(K,L,M,N)$ multi-access caching system with $\frac{M}{N}=\frac{\gamma}{K}$, $\gamma\in[0:\lfloor\frac{K}{L}\rfloor]$, let $t=\gamma L$, if $(K-t+1)|K$ or $K-t=1$, a $K\times K$ array $\mathbf{P}$ is defined as follows:
\begin{equation}
\label{eq_con_div}
\mathbf{P}(j,k)=\begin{cases}
(\langle j-k\rangle_K, \langle k\rangle_{K-t+1}), \ \ \text{if} \ \langle j-k\rangle_K< \frac{K-t+1}{2}\\
(K-t+1-\langle j-k\rangle_K, \langle j\rangle_{K-t+1}), \ \ \text{if} \ \frac{K-t+1}{2} <\langle j-k\rangle_K\leq K-t\\
(\langle j-k\rangle_K, \langle k\rangle_{\frac{K-t+1}{2}}), \ \ \text{if} \ \langle j-k\rangle_K= \frac{K-t+1}{2}\\
*,\ \ \text{otherwise.}
\end{cases}
\end{equation}
\end{construction}
The consecutive cyclic placement (i.e., each subfile $W_{n,j}$ can be retrieved by $t$ users $U_j, U_{\langle j+1\rangle_K},$ $\ldots, U_{\langle j+t-1\rangle_K}$) implies that $\mathbf{P}(j,k)=*$ if and only if $\langle j-k\rangle_K>K-t$. For each vector in $\mathbf{P}$, the first coordinate represents the category number of the track that the position of the vector belongs to.
When $K=10, L=3, \gamma=2, t=\gamma L=6$, the array generated by Construction \ref{constr1} is exactly the array $\mathbf{P}$ in Fig \ref{fill}. For example, since $\langle 1-1\rangle_{10}=10>K-t=4$, we have $\mathbf{P}(1,1)=*$ from \eqref{eq_con_div}; since $\langle 2-1\rangle_{10}=1<\frac{K-t+1}{2}=\frac{5}{2}$, we have $\mathbf{P}(2,1)=(\langle 2-1\rangle_{10}, \langle 1\rangle_{5})=(1,1)$ from \eqref{eq_con_div}; since $\langle 4-1\rangle_{10}=3>\frac{K-t+1}{2}=\frac{5}{2}$ and $\langle 4-1\rangle_{10}=3<K-t=4$, we have $\mathbf{P}(4,1)=(10-6+1-\langle 4-1\rangle_{10}, \langle 4\rangle_{5})=(2,4)$ from \eqref{eq_con_div}.
\begin{theorem}
\label{th_div}
For a $(K,L,M,N)$ multi-access caching system with $\frac{M}{N}=\frac{\gamma}{K}$, $\gamma\in[0:\lfloor\frac{K}{L}\rfloor]$, let $t=\gamma L$, if $(K-t+1)|K$ or $K-t=1$, the array generated by Construction \ref{constr1} is a $\frac{2K}{K-t+1}$-$\left(K,K,t,\frac{(K-t)(K-t+1)}{2}\right)$ PDA under the consecutive cyclic placement, which leads to a multi-access coded caching scheme with the rate $R_{new}=\frac{(K-t)(K-t+1)}{2K}$ and subpacketization $F_{new}=K$.
\end{theorem}
For the detailed proof of Theorem \ref{th_div}, please refer to Appendix \ref{pr_th_div}.

\subsection{The case of $(K-t+1)\nmid K$ and $K-t>1$}
\begin{example}
\label{example3}
When $K=5, L=2, \gamma=1$, then $t=\gamma L=2$, we have $(K-t+1)\nmid K$ and $K-t>1$. We further divide each subfile into $g=2\lfloor\frac{K}{K-t+1}\rfloor=2$ packets, i.e., $W_{n,j}=\{W_{n,(1,j)},W_{n,(2,j)}\}$, then the user-retrieve array $\mathbf{U}$ under the consecutive cyclical placement (i.e., each subfile $W_{n,j}=\{W_{n,(1,j)},W_{n,(2,j)}\}$ can be retrieved by $t=2$ users: $U_j, U_{\langle j+1\rangle_K}$) is shown in Fig \ref{fill3}.
Similar to Example \ref{example2}, we will fill in two-dimensional vectors at the non-star positions in the user-retrieve array $\mathbf{U}$, such that the resulting array is a PDA, as illustrated in Fig \ref{fill3}.
\begin{figure}
  \centering
  \includegraphics[width=5in]{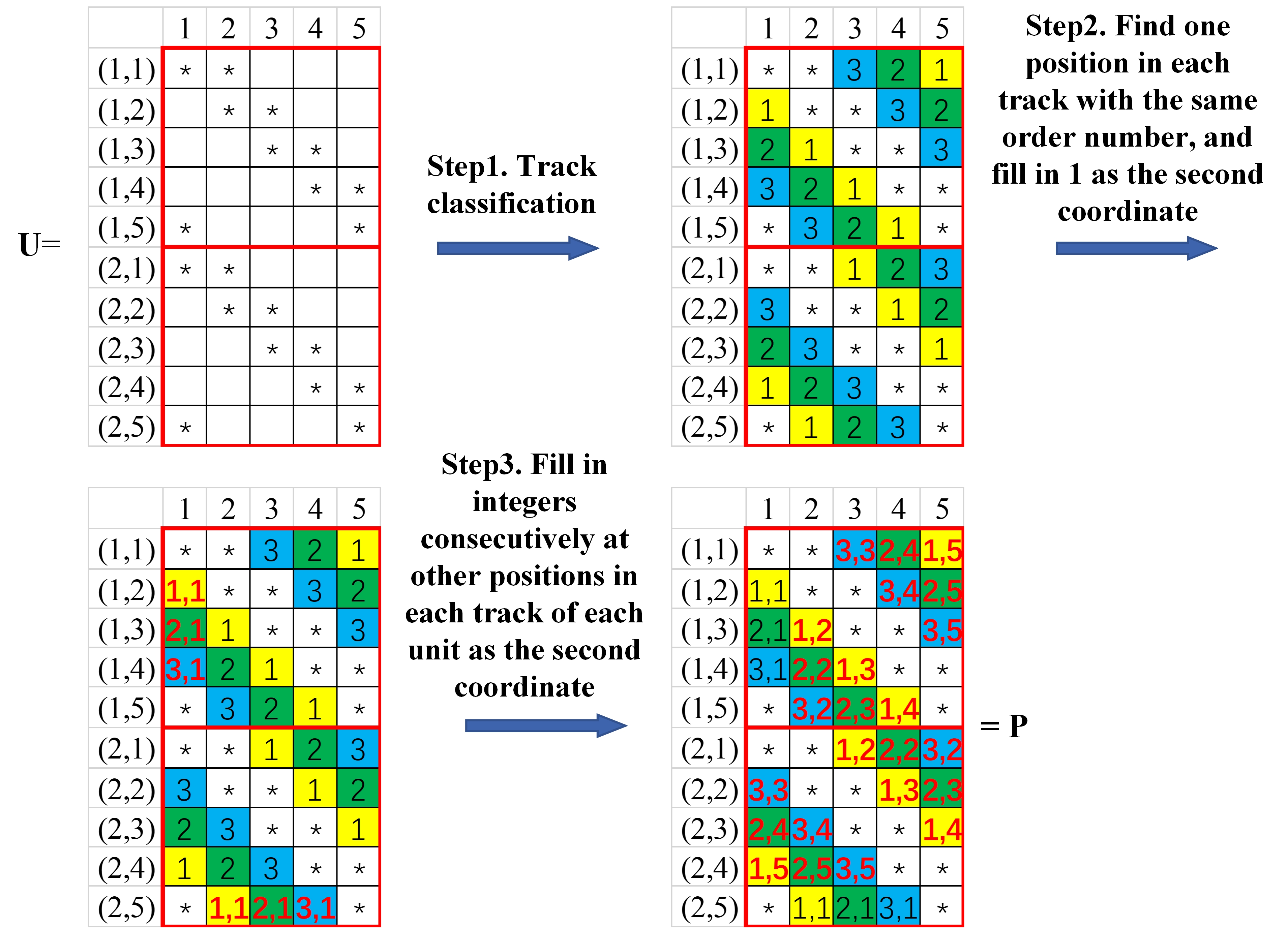}\\
  \caption{The steps of filling in two-dimensional vectors at the non-star positions in the user-retrieve array $\mathbf{U}$ when $K=5, L=2, \gamma=1$.}\label{fill3}
\vspace{-0.8cm}
\end{figure}

{\bf Step1.} We divide the user-retrieve array $\mathbf{U}$ into $g=2$ units, the first unit includes the rows indexed by $(1,j)$ and the second unit includes the rows indexed by $(2,j)$ where $j\in[1:7]$. Then we classify the non-star positions in each unit into $K-t=3$ tracks. For the first unit, the $l^{\text{th}}$ track is the collection of the positions $((1,j),k)$ satisfying $\langle j-k\rangle_K=l$; for the second unit, the $l^{\text{th}}$ track is the collection of the positions $((2,j),k)$ satisfying $K-t+1-\langle j-k\rangle_K=l$, where $l\in[1:3]$. The order number of each track is written into each position belonging to that track as the first coordinate of the two-dimensional vector.

{\bf Step2.} For each $g=2$ tracks with the same order number, we first find $g=2$ positions (in different tracks) in the order of the slowest increase of column labels, such that the element at the intersection of the row of any position and the column of another position is a star. Then we fill in $1$ at the found positions as the second coordinate. For example, for the first track in each unit (which are painted yellow in Fig \ref{fill3}), we find two positions in turn: $((1,2),1)\rightarrow ((2,5),2)$, which form the subarray in \eqref{subarray2}, and we fill in $1$ at the two non-star positions in \eqref{subarray2} as the second coordinate, then the vector $(1,1)$ appears $g=2$ times and satisfies Condition C3 of Definition \ref{def-PDA}.
\begin{equation}
\label{subarray2}
\bordermatrix{%
      & 1 & 2  \cr
(1,2) & 1 & *  \cr
(2,5) & * & 1}.
\end{equation}

{\bf Step3.} We fill in integers consecutively at other positions (whose second coordinate has not been filled in yet) in each track of each unit as the second coordinate. Precisely, if the vector at the position $((i,j),k)$ is $(l,1)$, then the vector at the position $((i,\langle j+\alpha\rangle_K), \langle k+\alpha\rangle_K)$ in the same track should be $(l,1+\alpha)$, where $\alpha\in[1:K-1]=[1:6]$.

It is easy to verify that the resulting array $\mathbf{P}$ is a $2$-$(5,10,4,15)$ PDA under the consecutive cyclic placement, which leads to a multi-access coded caching scheme with the rate $R=\frac{3}{2}$ and subpacketization $F=10$.
\end{example}

In general, the mathematical representation of the above construction is as follows:
\begin{construction}
\label{constr2}
For a $(K,L,M,N)$ multi-access caching system with $\frac{M}{N}=\frac{\gamma}{K}$, $\gamma\in[0:\lfloor\frac{K}{L}\rfloor]$, let $t=\gamma L$, if $(K-t+1)\nmid K$ and $K-t>1$, let
\begin{equation}
\label{eq_g}
g_{new}=\begin{cases}
2\lfloor\frac{K}{K-t+1}\rfloor+1, \ \ \text{if} \ \ \langle K\rangle_{K-t+1}=K-t\\ 2\lfloor\frac{K}{K-t+1}\rfloor, \ \ \text{otherwise}.
\end{cases}
\end{equation}
Let the row index set be $\mathcal{F}=[1:g_{new}]\times [1:K]$ and the column index set be $\mathcal{K}=[1:K]$, a $g_{new}K\times K$ array $\mathbf{P}=(\mathbf{P}((i,j),k))_{(i,j)\in\mathcal{F},k\in\mathcal{K}}$ is defined as follows:
\begin{small}
\begin{equation}
\label{eq_con}
\mathbf{P}((i,j),k)=\begin{cases}
\left(\langle j-k \rangle_K, \langle k-\frac{i-1}{2}(K-t+1)\rangle_K\right), \ \text{if} \ i \ \text{is odd and} \ \langle j-k\rangle_K\leq K-t\\
\left( K-t+1-\langle j-k \rangle_K, \langle j-\frac{i}{2}(K-t+1)\rangle_K\right), \ \text{if} \ i \ \text{is even and} \ \langle j-k\rangle_K\leq K-t\\
*,\ \ \ \ \ \text{otherwise.}
\end{cases}
\end{equation}
\end{small}
\end{construction}

If $(K-t+1)\nmid K$ and $K-t>1$, we need to further divide each subfile $W_{n,j}$ into $g_{new}$ packets, i.e., $W_{n,j}=\{W_{n,(1,j)},W_{n,(2,j)},\ldots,W_{n,(g_{new},j)}\}$, so the row index set is $[1:g_{new}]\times[1:K]$. The consecutive cyclical placement (where each subfile $W_{n,j}$ can be retrieved by $t$ users $U_j, U_{\langle j+1\rangle_K}, \ldots, U_{\langle j+t-1\rangle_K}$) implies that $\mathbf{P}((i,j),k)=*$ if and only if $\langle j-k\rangle_K>K-t$. For each vector in $\mathbf{P}$, the first coordinate represents the order number of the track that the position of the vector belongs to. When $K=5, L=2, \gamma=1, t=\gamma L=2$, the array generated by Construction \ref{constr2} is exactly the array $\mathbf{P}$ in Fig \ref{fill3}. For example, for $\mathbf{P}((1,1),1)$, since $\langle 1-1\rangle_5=5>K-t=3$, we have $\mathbf{P}((1,1),1)=*$ from \eqref{eq_con}; for $\mathbf{P}((1,2),1)$, since $1$ is odd and $\langle 2-1\rangle_5=1<K-t=3$, we have $\mathbf{P}((1,2),1)=(\langle 2-1\rangle_5, \langle 1-\frac{1-1}{2}(5-2+1)\rangle_5)=(1,1)$ from \eqref{eq_con}; for $\mathbf{P}((2,2),1)$, since $2$ is even and $\langle 2-1\rangle_5=1<K-t=3$, we have $\mathbf{P}((2,2),1)=(5-2+1-\langle 2-1\rangle_5, \langle 2-\frac{2}{2}(5-2+1)\rangle_5)=(3,3)$ from \eqref{eq_con}.

\begin{theorem}
\label{th_notdiv}
For a $(K,L,M,N)$ multi-access caching system with $\frac{M}{N}=\frac{\gamma}{K}$, $\gamma\in[0:\lfloor\frac{K}{L}\rfloor]$, let $t=\gamma L$, if $(K-t+1)\nmid K$ and $K-t>1$, the array generated by Construction \ref{constr2} is a $g_{new}$-$(K,g_{new}K,g_{new}t,K(K-t))$ PDA under the consecutive cyclical placement, which leads to a multi-access coded caching scheme with the rate $R_{new}=\frac{K-t}{g_{new}}$ and subpacketization $F_{new}=g_{new}K$, where $g_{new}$ is defined by \eqref{eq_g}.
\end{theorem}
For the detailed proof of Theorem \ref{th_notdiv}, please refer to Appendix \ref{pr_th_notdiv}.

\subsection{Proof of Theorem \ref{ach_R_F}}
\label{pr_th_ach_R_F}
For a $(K,L,M,N)$ multi-access caching system with $\frac{M}{N}=\frac{\gamma}{K}$, $\gamma\in[0:\lfloor\frac{K}{L}\rfloor]$, let $t=\gamma L$, if $(K-t+1)|K$ or $K-t=1$, the scheme in Theorem \ref{th_div} achieves the rate $R_{new}=\frac{(K-t)(K-t+1)}{2K}$ with subpacketization $F_{new}=K$ under the consecutive cyclical placement; if $\langle K\rangle_{K-t+1}=K-t$ and $K-t>1$, the scheme in Theorem \ref{th_notdiv} achieves the rate $R_{new}=\frac{K-t}{g_{new}}=\frac{K-t}{2\lfloor\frac{K}{K-t+1}\rfloor+1}$ with subpacketization $F_{new}=g_{new}K=(2\lfloor\frac{K}{K-t+1}\rfloor+1)K$ under the consecutive cyclical placement, since in this case $g_{new}=2\lfloor\frac{K}{K-t+1}\rfloor+1$ from \eqref{eq_g}; if $\langle K\rangle_{K-t+1}<K-t$ and $K-t>1$, the scheme in Theorem \ref{th_notdiv} achieves the rate $R_{new}=\frac{K-t}{g_{new}}=\frac{K-t}{2\lfloor\frac{K}{K-t+1}\rfloor}$ with subpacketization $F_{new}=g_{new}K=2\lfloor\frac{K}{K-t+1}\rfloor K$ under the consecutive cyclical placement, since in this case $g_{new}=2\lfloor\frac{K}{K-t+1}\rfloor$ from \eqref{eq_g}. The proof is complete.

\section{Conclusion}
\label{conclusion}
In this paper, we consider the $(K,L,M,N)$ multi-access caching system. First we propose the consecutive cyclic placement, which achieves the maximal local caching gain. Under the consecutive cyclic placement, we derive the optimal (maximal) coded caching gain from the perspective of PDA, thus obtaining a lower bound on the rate of PDA. Finally, we construct a class of PDA under the consecutive cyclic placement, which generates a multi-access coded caching scheme achieving our derived lower bound for some parameters;  while for other parameters, the achieved coded caching gain is only 1 less than the optimal one. Moreover, the subpacketization of the proposed scheme is less than $K^2$. Compared to some existing schemes with exponential subpacketization, our scheme has a lower subpacketization and simultaneously a lower rate. Compared to some existing schemes with linear subpacketization, our scheme has a better coded caching gain, thus achieving a lower rate.

\appendices

\section{Proof of Theorem \ref{maxgain}}
\label{pr_Theorem_Rlowerbound}
\begin{proof}
In order to prove Theorem \ref{maxgain}, we first prove the following proposition.
\begin{proposition}
\label{relationK_g}
For a $(K,L,M,N)$ multi-access caching system with $\frac{M}{N}=\frac{\gamma}{K}$, $\gamma\in[0:\lfloor\frac{K}{L}\rfloor]$, let $t=\gamma L$, if the maximal coded caching gain of a PDA under the consecutive cyclic placement is $g$, then $(g-2)K\leq g(t-1)$.
\end{proposition}
\begin{proof}
If $g\leq 2$, the proposition holds obviously. If $g\geq 3$, since the maximal coded caching gain is $g$, there exists an integer $s$ which appears $g$ times in $\mathbf{P}$. Assume that $\mathbf{P}(j_1,k_1)=\mathbf{P}(j_2,k_2)=\ldots=\mathbf{P}(j_g,k_g)=s$ where $j_1<j_2<\ldots<j_g$, then the subarray formed by rows $j_1,j_2,\ldots,j_g$ and columns $k_1,k_2,\ldots,k_g$ is equivalent to the following array
\begin{equation}
\label{eqsubarray}
\bordermatrix{%
&k_1 & k_2 & \ldots & k_{g}  \cr
j_1    & s & *         & \ldots & *             \cr
j_2    & *         & s & \ldots & *             \cr
\vdots & \vdots    & \vdots    & \ddots & \vdots    \cr
j_{g}& *         & *         & \ldots &s }
\end{equation}
from Condition C3 of Definition \ref{def-PDA}.
Since each subfile $W_{n,j}$ can be retrieved by $t$ users $U_j,U_{\langle j+1\rangle_K},$ $\ldots,U_{\langle j+t-1\rangle_K}$ from Definition \ref{conti_cyc}, we have $k_u\in \underset{v\in[1:g]\setminus \{u\}}{\bigcap}[j_v:j_v+t-1]_K$ for any $u\in[1:g]$.

For example, when $g=3$, we have $|[j_u:j_u+t-1]_K\cap[j_v:j_v+t-1]_K|\geq 1$ for any $u\neq v\in[1:3]$. Moreover, if $|[j_u:j_u+t-1]_K\cap[j_v:j_v+t-1]_K|=1$ for any $u\neq v\in[1:3]$, then $K=3t-3$, as illustrated in Fig. \ref{figmaxK}(a). Otherwise, we have $K<3t-3$, as illustrated in Fig. \ref{figmaxK}(b). Hence, when $g=3$, we have $K\leq 3t-3$. In general, if $|\underset{v\in[1:g]\setminus \{u\}}{\bigcap}[j_v:j_v+t-1]_K|=1$ for any $u\in[1:g]$, the value of $K$ is maximal, i.e., $K=\frac{gt-g}{g-2}$, since each element in $[1:K]$ is counted $g-2$ times to get $gt-g$. So we have $(g-2)K\leq g(t-1)$.
\begin{figure}[http!]
    \centering
    \subfigure[$K=3t-3$]{
        \includegraphics[width=2.8in]{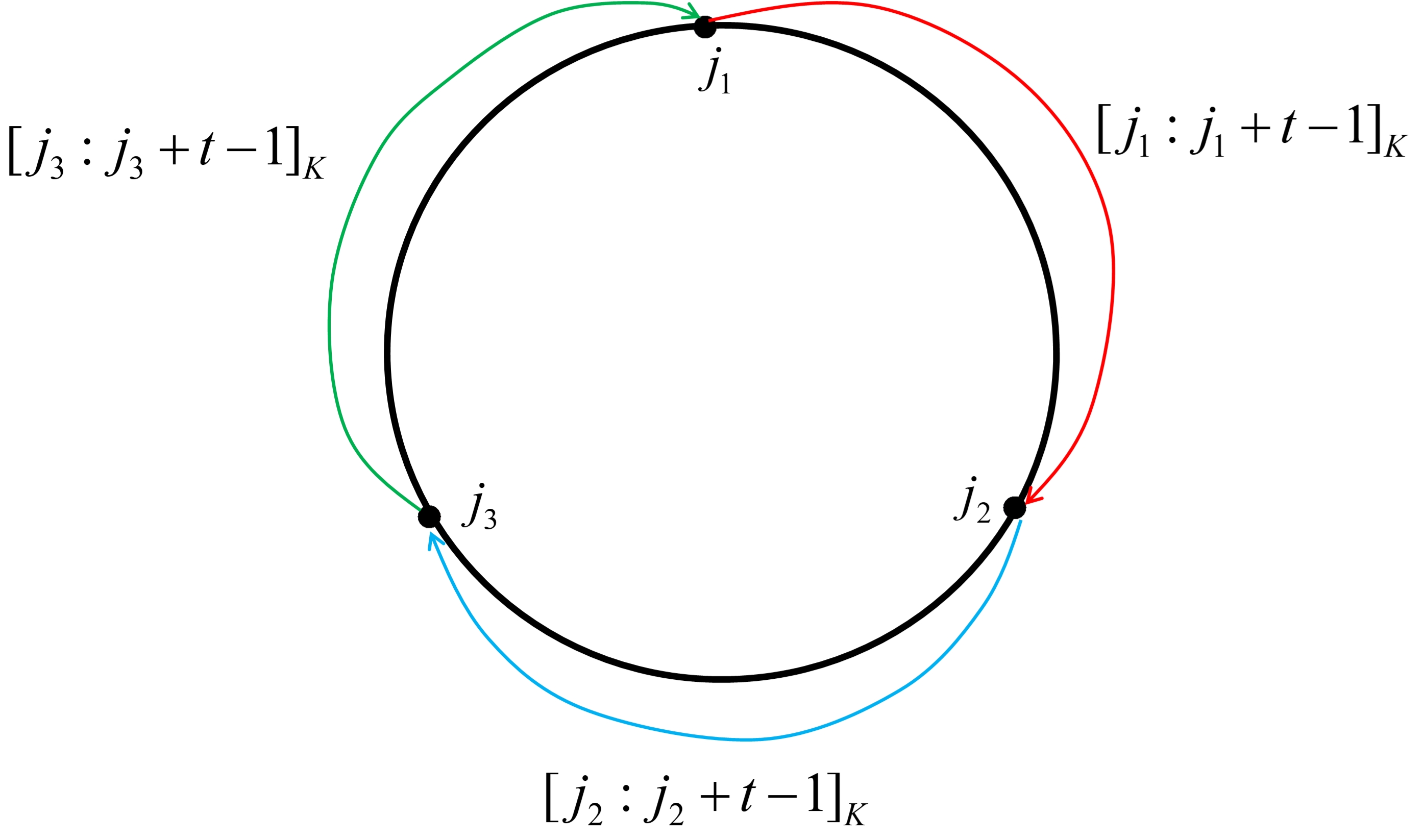}
    }
    ~
    \subfigure[$K<3t-3$]{
        \includegraphics[width=2.6in]{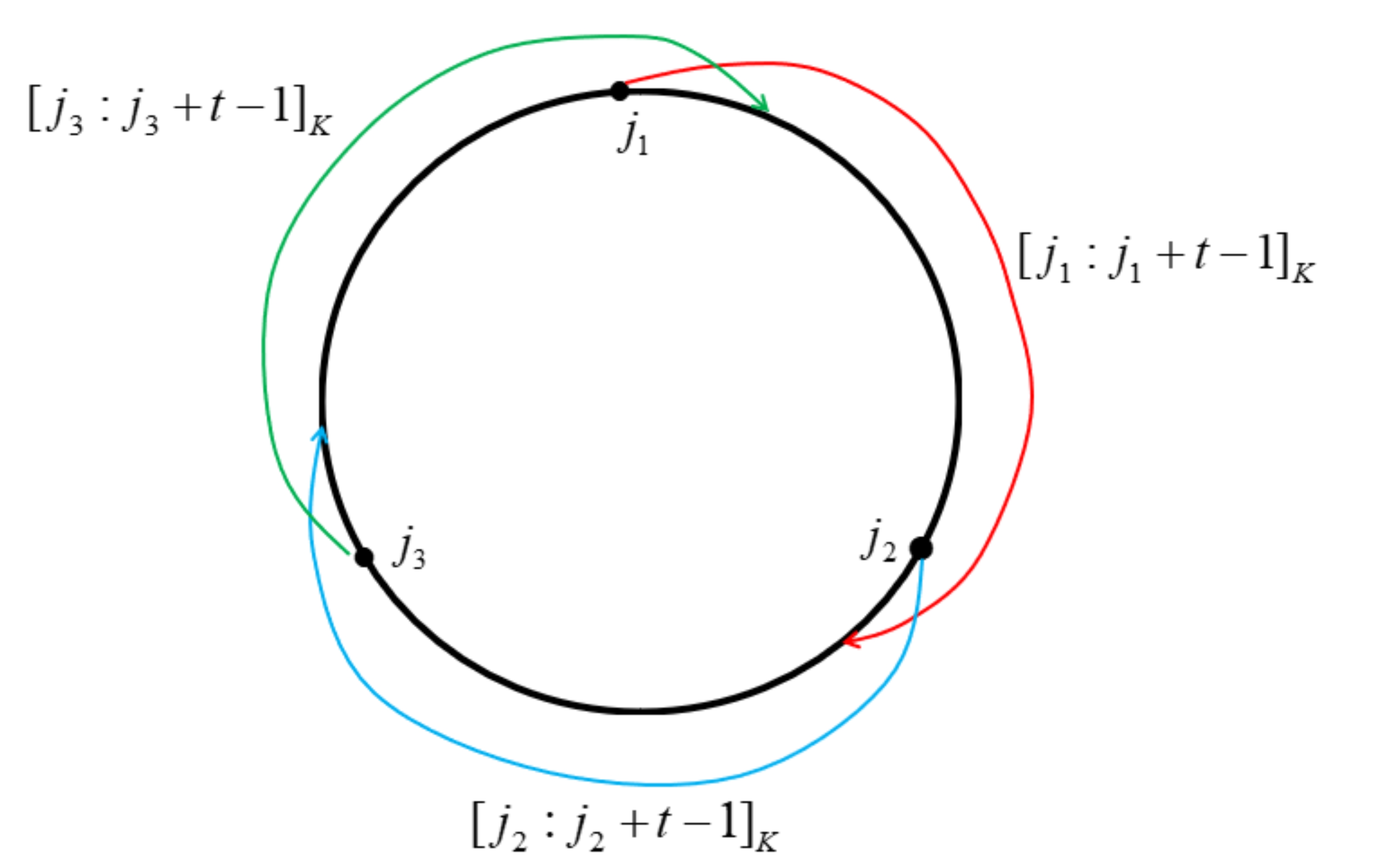}
    }
    \caption{The whole circle represents $[1:K]$, the ranges indicated by the red, blue, and green arrows represent $[j_1:j_1+t-1]_K$, $[j_2:j_2+t-1]_K$, and $[j_3:j_3+t-1]_K$ respectively.}
    \label{figmaxK}
\end{figure}
\end{proof}

Next we will prove Theorem \ref{maxgain}. Assume that the maximal coded caching gain of a PDA under the consecutive cyclic placement is $g$. If $t=0$, we have $g=1$, then the rate $R=K-t$. If $t\geq 1$, let $\lfloor\frac{K}{K-t+1}\rfloor=n$, then $n\geq 1$ and $g\geq 2$. If $\langle K\rangle_{K-t+1}\leq \lfloor\frac{K-t}{2}\rfloor$ or $(K-t+1)|K$, we have $K\leq n(K-t+1)+\frac{K-t}{2}$, which implies that $K\geq \frac{2n(t-1)+t}{2n-1}$. From Proposition \ref{relationK_g}, we have $(g-2)\frac{2n(t-1)+t}{2n-1}\leq (g-2)K\leq g(t-1)$, which implies that $g\leq 2n+\frac{2t-2n}{2t-1}$. Since $2t-2n<2t-1$, we have $g\leq 2n=2\lfloor\frac{K}{K-t+1}\rfloor$, then the rate $R\geq \frac{K-t}{g}\geq \frac{K-t}{2\lfloor\frac{K}{K-t+1}\rfloor}$. If $\lfloor\frac{K-t}{2}\rfloor<\langle K\rangle_{K-t+1}\leq K-t$, we have $K\leq n(K-t+1)+K-t$, which implies that $K\geq \frac{nt-n+t}{n}$. From Proposition \ref{relationK_g}, we have $(g-2)\frac{nt-n+t}{n}\leq (g-2)K\leq g(t-1)$, which implies that $g\leq 2n+1+\frac{t-2n}{t}$. Since $\frac{t-2n}{t}<1$, we have $g\leq 2n+1=2\lfloor\frac{K}{K-t+1}\rfloor+1$, then the rate $R\geq \frac{K-t}{g}\geq \frac{K-t}{2\lfloor\frac{K}{K-t+1}\rfloor+1}$.
\end{proof}

\section{Proof of Theorem \ref{th_div}}
\label{pr_th_div}
\begin{proof}
For a $(K,L,M,N)$ multi-access caching system with $\frac{M}{N}=\frac{\gamma}{K}$, $\gamma\in[0:\lfloor\frac{K}{L}\rfloor]$ and $t=\gamma L$, if $(K-t+1)|K$ or $K-t=1$, we will prove that the array $\mathbf{P}$ generated by Construction \ref{constr1} is a $\frac{2K}{K-t+1}$-$\left(K,K,t,\frac{(K-t)(K-t+1)}{2}\right)$ PDA under the consecutive cyclic placement.

1) For any $j,k\in[1:K]$, from \eqref{eq_con_div} we have
\begin{equation}
\label{eq_star_div}
\mathbf{P}(j,k)=* \Leftrightarrow K-t<\langle j-k\rangle_K\leq K,
\end{equation}
which satisfies the consecutive cyclic placement, i.e., each subfile $W_{n,j}$ can be retrieved by $t$ users: $U_j, U_{\langle j+1\rangle_K}, \ldots, U_{\langle j+t-1\rangle_K}$.
From \eqref{eq_star_div}, there are $t$ stars in each column of $\mathbf{P}$. Condition C1 of Definition \ref{def-PDA} holds.

2) We will prove that there are $\frac{(K-t)(K-t+1)}{2}$ different vectors in $\mathbf{P}$ and each vector appears exactly $\frac{2K}{K-t+1}$ times.
If $K-t=1$, we have $\mathbf{P}(j,k)\neq *$ if and only if $\langle j-k\rangle_K=1=\frac{K-t+1}{2}$ from \eqref{eq_star_div}. From \eqref{eq_con_div}, we have $\mathbf{P}(j,k)=(1,1)$ for any position $(j,k)$ satisfying $\langle j-k\rangle_K=1$. That is, there is only one vector in $\mathbf{P}$, which appears $K$ times. If $(K-t+1)|K$, let $n=\frac{K}{K-t+1}$.
    \begin{itemize}
    \item If $K-t$ is even, for any $s_1\in[1:\frac{K-t}{2}]$ and $s_2\in[1:K-t+1]$, from \eqref{eq_con_div} we have
    $\mathbf{P}(\langle s_1+s_2+i(K-t+1)\rangle_K,s_2+i(K-t+1))=(s_1,s_2)$
    for any $i\in[0:n-1]$, and
    $\mathbf{P}(s_2+j(K-t+1),\langle s_1+s_2+(j-1)(K-t+1)\rangle_K)=(s_1,s_2)$
    for any $j\in[0:n-1]$.
    Furthermore, we have $(\langle s_1+s_2+i(K-t+1)\rangle_K,s_2+i(K-t+1))\neq(s_2+j(K-t+1),\langle s_1+s_2+(j-1)(K-t+1)\rangle_K)$ for any $i,j\in[0:n-1]$, then the vector $(s_1, s_2)$ appears at least $2n$ times in $\mathbf{P}$. Otherwise if
    \begin{equation}
    \label{s1s2}
    \begin{array}{l}(\langle s_1+s_2+i(K-t+1)\rangle_K,s_2+i(K-t+1))=\\ \ \ \ \ \ \ \ (s_2+j(K-t+1),\langle s_1+s_2+(j-1)(K-t+1)\rangle_K),
    \end{array}
    \end{equation}
    then if $i=j$, we have $\langle s_1+s_2\rangle_K=s_2$ from \eqref{s1s2}, leading to $s_1=K$, which contradicts the hypothesis of $s_1\in[1:\frac{K-t}{2}]$;
    if $i<j$, from \eqref{s1s2} we have $s_1=\langle (j-i)(K-t+1)\rangle_K=(j-i)(K-t+1)\geq K-t+1>\frac{K-t}{2}$ (since $0<(j-i)(K-t+1)\leq(n-1)(K-t+1)<K$), which contradicts the hypothesis of $s_1\in[1:\frac{K-t}{2}]$; if $i>j$, from \eqref{s1s2} we have $s_1=\langle (i-j+1)(K-t+1)\rangle_K=(i-j+1)(K-t+1)>\frac{K-t}{2}$ (since $0<(i-j+1)(K-t+1)\leq n(K-t+1)=K$), which also contradicts the hypothesis of $s_1\in[1:\frac{K-t}{2}]$.
    On the other hand, since the average occurrence number of each vector in $[1:\frac{K-t}{2}]\times [1:K-t+1]$ is at most $\frac{K(K-t)}{\frac{K-t}{2}(K-t+1)}=2n$, there are $\frac{(K-t)(K-t+1)}{2}$ different vectors in $\mathbf{P}$ and each vector appears exactly $2n=\frac{2K}{K-t+1}$ times.

    \item If $K-t$ is odd, for any $s_1\in[1:\frac{K-t-1}{2}]$ and $s_2\in[1:K-t+1]$, following the similar process in the case that $K-t$ is even, we have that the vector $(s_1, s_2)$ appears at least $2n$ times in $\mathbf{P}$. In addition, for any $s\in[1:\frac{K-t+1}{2}]$, from \eqref{eq_con_div} we have $\mathbf{P}\left(\left\langle s+(i+1)\frac{K-t+1}{2}\right\rangle_K, s+i\frac{K-t+1}{2}\right)$ $=\left(\frac{K-t+1}{2},s\right)$ for any $i\in[0:2n-1]$. Hence, each vector in $([1:\frac{K-t-1}{2}]\times [1:K-t+1])\cup (\{\frac{K-t+1}{2}\}\times [1:\frac{K-t+1}{2}])$ appears at least $2n$ times in $\mathbf{P}$.
    On the other hand, since the average occurrence number of each vector in $([1:\frac{K-t-1}{2}]\times [1:K-t+1])\cup (\{\frac{K-t+1}{2}\}\times [1:\frac{K-t+1}{2}])$ is at most $\frac{K(K-t)}{\frac{K-t-1}{2}(K-t+1)+\frac{K-t+1}{2}}=2n$, there are $\frac{K-t-1}{2}(K-t+1)+\frac{K-t+1}{2}=\frac{(K-t)(K-t+1)}{2}$ different vectors in $\mathbf{P}$ and each vector appears exactly $2n=\frac{2K}{K-t+1}$ times.

    \end{itemize}
    Condition C2 of Definition \ref{def-PDA} holds.

3) For any two distinct entries, say $\mathbf{P}(j_1,k_1)$ and $\mathbf{P}(j_2,k_2)$, if $\mathbf{P}(j_1,k_1)=\mathbf{P}(j_2,k_2)=(s_1, s_2)$, we will prove that $\mathbf{P}(j_2,k_1)=\mathbf{P}(j_1,k_2)=*$.
If $K-t=1$, we have $\langle j_1-k_1\rangle_K=\langle j_2-k_2\rangle_K=1$ from \eqref{eq_star_div}. If $\langle j_2-k_1\rangle_K=1$, we have $j_1=j_2$ and $k_1=k_2$, which contradicts the hypothesis that $\mathbf{P}(j_1,k_1)$ and $\mathbf{P}(j_2,k_2)$ are two distinct entries. Hence, we have $\langle j_2-k_1\rangle_K\neq1$, which implies $\mathbf{P}(j_2,k_1)=*$ from \eqref{eq_star_div}.
If $(K-t+1)|K$, let $n=\frac{K}{K-t+1}$.
    \begin{itemize}
    \item If $\langle j_1-k_1\rangle_K<\frac{K-t+1}{2}$, then $\langle j_2-k_2\rangle_K\neq\frac{K-t+1}{2}$. Otherwise if $\langle j_2-k_2\rangle_K=\frac{K-t+1}{2}$, we have $s_1=\langle j_1-k_1\rangle_K=\langle j_2-k_2\rangle_K=\frac{K-t+1}{2}$ from \eqref{eq_con_div}, which contradicts the hypothesis of $\langle j_1-k_1\rangle_K<\frac{K-t+1}{2}$.
        \begin{itemize}
        \item If $\langle j_2-k_2\rangle_K<\frac{K-t+1}{2}$, from \eqref{eq_con_div} we have
        \begin{align}
        s_1=\langle j_1-k_1\rangle_K=\langle j_2-k_2\rangle_K, \label{s1_div}\\
        s_2=\langle k_1\rangle_{K-t+1}=\langle k_2\rangle_{K-t+1}.\ \label{s2_div}
        \end{align}
        Then we have $k_1\neq k_2$. Otherwise if $k_1=k_2$, we have $j_1=j_2$ from \eqref{s1_div}, which contradicts the hypothesis that $\mathbf{P}(j_1,k_1)$ and $\mathbf{P}(j_2,k_2)$ are two distinct entries. If $k_1<k_2$, from \eqref{s2_div} there exists some $i\in[1:n-1]$, such that
        \begin{equation}
        \label{k1k2}
        k_2-k_1=i(K-t+1).
        \end{equation}
        Then we have
        \begin{equation}
        \label{j2k1_div}
        \langle j_2-k_1\rangle_K\overset{\eqref{s1_div}}{=}\langle j_1-k_1+k_2-k_1\rangle_K\overset{\eqref{k1k2}}{=}\langle j_1-k_1+i(K-t+1)\rangle_K.
        \end{equation}
        Since $\langle j_1-k_1\rangle_K<\frac{K-t+1}{2}$ implies $1\leq j_1-k_1<\frac{K-t+1}{2}$ or $1\leq K+j_1-k_1<\frac{K-t+1}{2}$. If $1\leq j_1-k_1<\frac{K-t+1}{2}$ holds, we have
        $K-t+2\leq j_1-k_1+i(K-t+1)<\frac{K-t+1}{2}+(n-1)(K-t+1)<n(K-t+1)=K$ since $i\in[1:n-1]$ and $K=n(K-t+1)$. Then from \eqref{j2k1_div} we have $\langle j_2-k_1\rangle_K=\langle j_1-k_1+i(K-t+1)\rangle_K>K-t$, which implies $\mathbf{P}(j_2,k_1)=*$ from \eqref{eq_star_div}. If $1\leq K+j_1-k_1<\frac{K-t+1}{2}$ holds, from \eqref{j2k1_div} we have $\langle j_2-k_1\rangle_K=\langle K+j_1-k_1+i(K-t+1)\rangle_K>K-t$ similarly, which also implies $\mathbf{P}(j_2,k_1)=*$ from \eqref{eq_star_div}.

        \item If $\frac{K-t+1}{2}<\langle j_2-k_2\rangle_K\leq K-t$, from \eqref{eq_con_div} we have
        \begin{equation}
        \label{s2_div1}
        s_2=\langle k_1\rangle_{K-t+1}=\langle j_2\rangle_{K-t+1}.
        \end{equation}
        If $j_2\geq k_1$, there exists $i\in[0:n-1]$, such that $j_2-k_1=i(K-t+1)$, which implies $\langle j_2-k_1\rangle_K>K-t$, leading to $\mathbf{P}(j_2,k_1)=*$ from \eqref{eq_star_div}; if $j_2< k_1$, there exists $i\in[1:n-1]$, such that $k_1-j_2=i(K-t+1)$, which implies $\langle j_2-k_1\rangle_K=\langle K-i(K-t+1)\rangle_K>K-t$ (since $K-t+1=K-(n-1)(K-t+1)\leq K-i(K-t+1)\leq K$), leading to $\mathbf{P}(j_2,k_1)=*$ from \eqref{eq_star_div}.
        \end{itemize}
    \item If $\frac{K-t+1}{2}<\langle j_1-k_1\rangle_K\leq K-t$, we also have $\langle j_2-k_2\rangle_K\neq\frac{K-t+1}{2}$.
        \begin{itemize}
        \item If $\langle j_2-k_2\rangle_K<\frac{K-t+1}{2}$, from \eqref{eq_con_div} we have
        \begin{align}
        s_1=K-t+1-\langle j_1-k_1\rangle_K=\langle j_2-k_2\rangle_K, \label{s1_div2}\\
        s_2=\langle j_1\rangle_{K-t+1}=\langle k_2\rangle_{K-t+1}. \ \ \ \ \ \ \ \ \label{s2_div2}
        \end{align}
        If $j_1=k_2$, from \eqref{s1_div2} we have $\langle j_2-k_1\rangle_K=K-t+1>K-t$, which implies $\mathbf{P}(j_2,k_1)=*$ from \eqref{eq_star_div}. If $j_1>k_2$, from \eqref{s2_div2}, there exists $i\in[1:n-1]$ such that $j_1-k_2=i(K-t+1)$, then we have
        $$\langle j_2-k_1\rangle_K\overset{\eqref{s1_div2}}{=}\langle K-t+1-(j_1-k_2)\rangle_K=\langle K+K-t+1-i(K-t+1)\rangle_K>K-t,$$ since $2(K-t+1)\leq K+K-t+1-i(K-t+1)\leq K$. Consequently, we have $\mathbf{P}(j_2,k_1)=*$ from \eqref{eq_star_div}. If $j_1<k_2$, from \eqref{s2_div2}, there exists $i\in[1:n-1]$ such that $k_2-j_1=i(K-t+1)$, then we have
        $$\langle j_2-k_1\rangle_K\overset{\eqref{s1_div2}}{=}\langle K-t+1+k_2-j_1\rangle_K=\langle K-t+1+i(K-t+1)\rangle_K>K-t,$$ since $2(K-t+1)\leq K-t+1+i(K-t+1)\leq K-t+1+(n-1)(K-t+1)=K.$ Consequently, we have $\mathbf{P}(j_2,k_1)=*$ from \eqref{eq_star_div}.

        \item If $\frac{K-t+1}{2}<\langle j_2-k_2\rangle_K\leq K-t$, from \eqref{eq_con_div} we have
        \begin{align}
        s_1=K-t+1-\langle j_1-k_1\rangle_K=K-t+1-\langle j_2-k_2\rangle_K, \label{s1_div3}\\
        s_2=\langle j_1\rangle_{K-t+1}=\langle j_2\rangle_{K-t+1}. \ \ \ \ \ \ \ \ \ \ \ \ \ \ \ \ \ \ \ \label{s2_div3}
        \end{align}
        Then we have $j_1\neq j_2$. Otherwise if $j_1=j_2$, we have $k_1=k_2$ from \eqref{s1_div3}, which contradicts the hypothesis that $\mathbf{P}(j_1,k_1)$ and $\mathbf{P}(j_2,k_2)$ are two distinct entries. If $j_1<j_2$, from \eqref{s2_div3}, there exists $i\in[1:n-1]$ such that $j_2-j_1=i(K-t+1)$, then we have
        \begin{equation}
        \label{j2k1_4}
        \langle j_2-k_1\rangle_K\overset{\eqref{s1_div3}}{=}\langle j_2-j_1+j_2-k_2\rangle_K=\langle i(K-t+1)+j_2-k_2\rangle_K.
        \end{equation}
        Since $\frac{K-t+1}{2}<\langle j_2-k_2\rangle_K\leq K-t$ implies $\frac{K-t+1}{2}<j_2-k_2\leq K-t$ or $\frac{K-t+1}{2}<K+j_2-k_2\leq K-t$. If $\frac{K-t+1}{2}<j_2-k_2\leq K-t$ holds, we have $\frac{3}{2}(K-t+1)<i(K-t+1)+j_2-k_2\leq (n-1)(K-t+1)+K-t<n(K-t+1)=K,$ then from \eqref{j2k1_4} we have $\langle j_2-k_1\rangle_K>K-t$, which implies $\mathbf{P}(j_2,k_1)=*$ from \eqref{eq_star_div}. If $\frac{K-t+1}{2}<K+j_2-k_2\leq K-t$ holds, from \eqref{j2k1_4} we have $\langle j_2-k_1\rangle_K=\langle i(K-t+1)+K+j_2-k_2\rangle_K>K-t$ similarly, which implies $\mathbf{P}(j_2,k_1)=*$ from \eqref{eq_star_div}. If $j_1>j_2$, from \eqref{s2_div3}, there exists $i\in[1:n-1]$ such that $j_1-j_2=i(K-t+1)$, then we have
        \begin{equation}
        \label{j2k1_5}
        \langle j_2-k_1\rangle_K\overset{\eqref{s1_div3}}{=}\langle j_2-j_1+j_2-k_2\rangle_K=\langle j_2-k_2-i(K-t+1)\rangle_K.
        \end{equation}
        If $\frac{K-t+1}{2}<j_2-k_2\leq K-t$ holds, we have $$
        \begin{array}{l}\frac{3}{2}(K-t+1)=K+\frac{K-t+1}{2}-(n-1)(K-t+1)\\ \ \ \ \  <K+j_2-k_2 -i(K-t+1)\leq K+K-t-(K-t+1)<K,\end{array}$$ then from \eqref{j2k1_5} we have $\langle j_2-k_1\rangle_K=\langle K+j_2-k_2 -i(K-t+1)\rangle_K>K-t$, which implies $\mathbf{P}(j_2,k_1)=*$ from \eqref{eq_star_div}. If $\frac{K-t+1}{2}<K+j_2-k_2\leq K-t$ holds, from \eqref{j2k1_5} we have $\langle j_2-k_1\rangle_K=\langle K+K+j_2-k_2-i(K-t+1)+K\rangle_K>K-t$ similarly, which implies $\mathbf{P}(j_2,k_1)=*$ from \eqref{eq_star_div}.
        \end{itemize}

    \item If $\langle j_1-k_1\rangle_K=\frac{K-t+1}{2}$, we have $\langle j_2-k_2\rangle_K=\frac{K-t+1}{2}$. Otherwise, if $\langle j_2-k_2\rangle_K<\frac{K-t+1}{2}$, from \eqref{eq_con_div} we have $s_1=\langle j_1-k_1\rangle_K=\langle j_2-k_2\rangle_K<\frac{K-t+1}{2}$, which contradicts the hypothesis of $\langle j_1-k_1\rangle_K=\frac{K-t+1}{2}$; if $\frac{K-t+1}{2}<\langle j_2-k_2\rangle_K\leq K-t$, from  \eqref{eq_con_div} we have $s_1=\langle j_1-k_1\rangle_K=K-t+1-\langle j_2-k_2\rangle_K<\frac{K-t+1}{2}$, which also contradicts the hypothesis of $\langle j_1-k_1\rangle_K=\frac{K-t+1}{2}$. Then from \eqref{eq_con_div} we have
        \begin{align}
        s_1=\langle j_1-k_1\rangle_K=\langle j_2-k_2\rangle_K=\frac{K-t+1}{2}, \label{s1_div4}\\
        s_2=\langle k_1\rangle_{\frac{K-t+1}{2}}=\langle k_2\rangle_{\frac{K-t+1}{2}}.\ \ \ \ \ \ \ \ \ \  \label{s2_div4}
        \end{align}
        Consequently, we have $k_1\neq k_2$. Otherwise if $k_1=k_2$, we have $j_1=j_2$ from \eqref{s1_div4}, which contradicts with the hypothesis that $\mathbf{P}(j_1,k_1)$ and $\mathbf{P}(j_2,k_2)$ are two distinct entries. If $k_1<k_2$, from \eqref{s2_div4}, there exists $i\in[1:2n-1]$ such that $k_2-k_1=i\frac{K-t+1}{2}$, then we have
        \begin{equation}
        \label{j2k1_6}
        \langle j_2-k_1\rangle_K\overset{\eqref{s1_div4}}{=}\langle j_1-k_1+k_2-k_1\rangle_K=\left\langle j_1-k_1+i\frac{K-t+1}{2}\right\rangle_K.
        \end{equation}
        Since $\langle j_1-k_1\rangle_K=\frac{K-t+1}{2}$ implies $j_1-k_1=\frac{K-t+1}{2}$ or $K+j_1-k_1=\frac{K-t+1}{2}$, if $j_1-k_1=\frac{K-t+1}{2}$ holds, we have
        $K-t+1\leq j_1-k_1+i\frac{K-t+1}{2}\leq \frac{K-t+1}{2}+(2n-1)\frac{K-t+1}{2}=K,$
        then from \eqref{j2k1_6} we have $\langle j_2-k_1\rangle_K>K-t$, which leads to $\mathbf{P}(j_2,k_1)=*$ from \eqref{eq_star_div}. If $K+j_1-k_1=\frac{K-t+1}{2}$ holds, we have $\langle j_2-k_1\rangle_K=\langle K+j_1-k_1+i\frac{K-t+1}{2}\rangle_K>K-t$ similarly, which leads to $\mathbf{P}(j_2,k_1)=*$ from \eqref{eq_star_div}.
        If $k_1>k_2$, from \eqref{s2_div4}, there exists $i\in[1:2n-1]$ such that $k_1-k_2=i\frac{K-t+1}{2}$, then we have
        \begin{equation}
        \label{j2k1_7}
        \langle j_2-k_1\rangle_K\overset{\eqref{s1_div4}}{=}\langle j_1-k_1+k_2-k_1\rangle_K=\left\langle j_1-k_1-i\frac{K-t+1}{2}\right\rangle_K.
        \end{equation}
        If $j_1-k_1=\frac{K-t+1}{2}$ holds, we have
        $$\begin{array}{l}K-t+1=K+\frac{K-t+1}{2}-(2n-1)\frac{K-t+1}{2}\leq K+j_1-k_1-i\frac{K-t+1}{2}\\ \ \ \ \ \ \ \ \ \ \ \ \ \ \leq K+\frac{K-t+1}{2}-\frac{K-t+1}{2}=K,\end{array}$$
        then from \eqref{j2k1_7} we have $\langle j_2-k_1\rangle_K=\langle K+j_1-k_1-i\frac{K-t+1}{2}\rangle_K>K-t$, which leads to $\mathbf{P}(j_2,k_1)=*$ from \eqref{eq_star_div}. If $K+j_1-k_1=\frac{K-t+1}{2}$ holds, we have $\langle j_2-k_1\rangle_K=\langle K+K+j_1-k_1-i\frac{K-t+1}{2}\rangle_K>K-t$ similarly, which leads to $\mathbf{P}(j_2,k_1)=*$ from \eqref{eq_star_div}.
    \end{itemize}
    Similarly, we can prove that $\mathbf{P}(j_1,k_2)=*$.
    Condition C3 of Definition \ref{def-PDA} holds.
Therefore, the array $\mathbf{P}$ generated by Construction \ref{constr1} is a $\frac{2K}{K-t+1}$-$\left(K,K,t,\frac{(K-t)(K-t+1)}{2}\right)$ PDA under the consecutive cyclic placement, which leads to a multi-access coded caching scheme with the rate $R=\frac{(K-t)(K-t+1)}{2K}$ and subpacketization $F=K$.
\end{proof}

\section{Proof of Theorem \ref{th_notdiv}}
\label{pr_th_notdiv}
\begin{proof}
For a $(K,L,M,N)$ multi-access caching system with $\frac{M}{N}=\frac{\gamma}{K}$, $\gamma\in[0:\lfloor\frac{K}{L}\rfloor]$ and $t=\gamma L$, if $(K-t+1)\nmid K$ and $K-t>1$, we will prove that the array $\mathbf{P}$ generated by Construction \ref{constr2} is a $g_{new}$-$(K,g_{new}K,g_{new}t,K(K-t))$ PDA under the consecutive cyclical placement.

1) For any $i\in[1:g_{new}]$ and any $j,k\in[1:K]$, from \eqref{eq_con} we have
\begin{equation}
\label{eq_star}
\mathbf{P}((i,j),k)=* \Leftrightarrow \langle j-k\rangle_K> K-t.
\end{equation}
Hence, there are $g_{new}t$ stars in each column of $\mathbf{P}$. Condition C1 of Definition \ref{def-PDA} holds.

2) For any $s_1\in[1:K-t]$ and any $s_2\in[1:K]$, we will prove that the vector $(s_1,s_2)$ appears exactly $g_{new}$ times in $\mathbf{P}$.
    For any $i\in[1:g_{new}]$, if $i$ is odd, let $j=\langle \frac{i-1}{2}(K-t+1)+s_1+s_2\rangle_K$ and $k=\langle \frac{i-1}{2}(K-t+1) +s_2\rangle_K$, then $\langle j-k\rangle_K=s_1\leq K-t$, so $\mathbf{P}((i,j),k)=(s_1,s_2)$ from \eqref{eq_con};
    if $i$ is even, let $j=\langle \frac{i}{2}(K-t+1)+s_2\rangle_K$ and $k=\langle (\frac{i}{2}-1)(K-t+1)+s_1+s_2\rangle_K$, then $\langle j-k\rangle_K=K-t+1-s_1\leq K-t$, so $\mathbf{P}((i,j),k)=(s_1,s_2)$ from \eqref{eq_con}.
    Hence, the vector $(s_1,s_2)$ appears at least $g_{new}$ times in $\mathbf{P}$. On the other hand, since the average occurrence number of each vector in $[1:K-t]\times[1:K]$ is no more than $\frac{K(g_{new}K-g_{new}t)}{K(K-t)}=g_{new}$, there are exactly $S=K(K-t)$ different vectors in $\mathbf{P}$, and each vector appears exactly $g_{new}$ times. Condition C2 of Definition \ref{def-PDA} holds.

3) For any two distinct entries, say $\mathbf{P}((i_1,j_1),k_1)$ and $\mathbf{P}((i_2,j_2),k_2)$, if $\mathbf{P}((i_1,j_1),k_1)=\mathbf{P}((i_2,j_2),k_2)=(s_1,s_2)$, we will prove that $\mathbf{P}((i_2,j_2),k_1)=\mathbf{P}((i_1,j_1),k_2)=*$.
    First we have $i_1\neq i_2$. Otherwise, if $i_1=i_2$ is odd, we have $s_1=\langle j_1-k_1\rangle_K=\langle j_2-k_2\rangle_K$ and $s_2=\langle k_1-\frac{i_1-1}{2}(K-t+1)\rangle_K=\langle k_2-\frac{i_2-1}{2}(K-t+1)\rangle_K$ from \eqref{eq_con}, leading to $k_1=k_2$ and $j_1=j_2$, which  contradicts the hypothesis that $\mathbf{P}((i_1,j_1),k_1)$ and $\mathbf{P}((i_2,j_2),k_2)$ are two distinct entries;
    if $i_1=i_2$ is even, we have $s_1= K-t+1-\langle j_1-k_1\rangle_K= K-t+1-\langle j_2-k_2\rangle_K$ and $s_2=\langle j_1-\frac{i_1}{2}(K-t+1)\rangle_K=\langle j_2-\frac{i_2}{2}(K-t+1)\rangle_K$ from \eqref{eq_con}, leading to $j_1=j_2$ and $k_1=k_2$, which also contradicts the hypothesis that $\mathbf{P}((i_1,j_1),k_1)$ and $\mathbf{P}((i_2,j_2),k_2)$ are two distinct entries.    
    Without loss of generality, we assume that $i_1<i_2$.
    \begin{itemize}
    \item If $i_1$ and $i_2$ are odd, we have
          \begin{align}
          &s_1=\langle j_1-k_1\rangle_K=\langle j_2-k_2\rangle_K\in[1:K-t] \label{s1_2odd} \\
          &s_2=\left\langle k_1-\frac{i_1-1}{2}(K-t+1)\right\rangle_K=\left\langle k_2-\frac{i_2-1}{2}(K-t+1)\right\rangle_K \label{s2_2odd}
          \end{align}
          from \eqref{eq_con}. Then we have
        \begin{equation}
        \label{j2k1_2odd}
        \langle j_2-k_1\rangle_K\overset{\eqref{s1_2odd}}{=}\langle j_1-k_1+k_2-k_1\rangle_K\overset{\eqref{s2_2odd}}{=}\left\langle j_1-k_1+\frac{i_2-i_1}{2}(K-t+1)\right\rangle_K.
        \end{equation}

      \begin{itemize}
        \item If $\langle K\rangle_{K-t+1}=K-t$, we have
        \begin{equation}
        \label{eq_g_odd}
        g_{new}=2\left\lfloor\frac{K}{K-t+1}\right\rfloor+1
        \end{equation}
        from \eqref{eq_g} and
        \begin{equation}
        \label{eq_modK}
        K-\left\lfloor\frac{K}{K-t+1}\right\rfloor(K-t+1)=K-t.
        \end{equation}
        From \eqref{s1_2odd} we have $1\leq j_1-k_1\leq K-t$ or $1\leq K+j_1-k_1\leq K-t$. If $1\leq j_1-k_1\leq K-t$ holds, we have
        \begin{align*}
        K-t+2&\leq j_1-k_1+\frac{i_2-i_1}{2}(K-t+1)\leq K-t+\frac{g_{new}-1}{2}(K-t+1) \\
        &\overset{\eqref{eq_g_odd}}{=}K-t+\left\lfloor\frac{K}{K-t+1}\right\rfloor(K-t+1)\overset{\eqref{eq_modK}}{=}K,
        \end{align*}
        since $i_1,i_2\in[1:g_{new}]$ and $i_1, i_2$ are odd. Therefore, from \eqref{j2k1_2odd} we have
        $\langle j_2-k_1\rangle_K=\langle j_1-k_1+\frac{i_2-i_1}{2}(K-t+1)\rangle_K> K-t$, which implies $\mathbf{P}((i_2,j_2),k_1)=*$ from \eqref{eq_star}. If $1\leq K+j_1-k_1\leq K-t$ holds, from \eqref{j2k1_2odd} we have
        $\langle j_2-k_1\rangle_K=\langle K+j_1-k_1+\frac{i_2-i_1}{2}(K-t+1)\rangle_K> K-t$ similarly, which also implies $\mathbf{P}((i_2,j_2),k_1)=*$ from \eqref{eq_star}.

        \item If $\langle K\rangle_{K-t+1}\neq K-t$, we have
        \begin{equation}
        \label{eq_g_even}
        g_{new}=2\left\lfloor\frac{K}{K-t+1}\right\rfloor.
        \end{equation}
        If $1\leq j_1-k_1\leq K-t$ holds, we have
        \begin{align*}
        K-t+2&\leq j_1-k_1+\frac{i_2-i_1}{2}(K-t+1)\leq K-t+\frac{g_{new}-1-1}{2}(K-t+1) \\
        &\overset{\eqref{eq_g_even}}{=}\left\lfloor\frac{K}{K-t+1}\right\rfloor(K-t+1)-1<K,
        \end{align*}
        since $i_1,i_2\in[1:g_{new}]$ and $i_1, i_2$ are odd. Therefore, from \eqref{j2k1_2odd} we have
        $\langle j_2-k_1\rangle_K=\langle j_1-k_1+\frac{i_2-i_1}{2}(K-t+1)\rangle_K> K-t$, which implies $\mathbf{P}((i_2,j_2),k_1)=*$ from \eqref{eq_star}. If $1\leq K+j_1-k_1\leq K-t$ holds, from \eqref{j2k1_2odd} we have
        $\langle j_2-k_1\rangle_K=\langle K+j_1-k_1+\frac{i_2-i_1}{2}(K-t+1)\rangle_K> K-t$ similarly, which also implies $\mathbf{P}((i_2,j_2),k_1)=*$ from \eqref{eq_star}.
      \end{itemize}

    \item If $i_1$ and $i_2$ are even, then we have
          \begin{align}
          &s_1=K-t+1-\langle j_1-k_1\rangle_K=K-t+1-\langle j_2-k_2\rangle_K\in[1:K-t] \label{s1_2even} \\
          &s_2=\left\langle j_1-\frac{i_1}{2}(K-t+1)\right\rangle_K=\left\langle j_2-\frac{i_2}{2}(K-t+1)\right\rangle_K \label{s2_2even}
          \end{align}
          from \eqref{eq_con}. Then we have
        \begin{equation}
        \label{j2k1_2even}
        \langle j_2-k_1\rangle_K\overset{\eqref{s1_2even}}{=}\langle j_2-k_2+j_2-j_1\rangle_K\overset{\eqref{s2_2even}}{=}\left\langle j_2-k_2+\frac{i_2-i_1}{2}(K-t+1)\right\rangle_K.
        \end{equation}
      \begin{itemize}
      \item If $\langle K\rangle_{K-t+1}= K-t$, we have \eqref{eq_g_odd} and \eqref{eq_modK}. From \eqref{s1_2even} we have $1\leq j_2-k_2\leq K-t$ or $1\leq K+j_2-k_2\leq K-t$. If $1\leq j_2-k_2\leq K-t$ holds, we have
        \begin{align*}
        K-t+2&\leq j_2-k_2+\frac{i_2-i_1}{2}(K-t+1)\leq K-t+\frac{g_{new}-1-2}{2}(K-t+1) \\
        &\overset{\eqref{eq_g_odd}}{=}\left\lfloor\frac{K}{K-t+1}\right\rfloor(K-t+1)-1< K,\\
        \end{align*}
        since $i_1,i_2\in[1:g_{new}]$ and $i_1, i_2$ are even. Therefore, from \eqref{j2k1_2even} we have
        $\langle j_2-k_1\rangle_K=\langle j_2-k_2+\frac{i_2-i_1}{2}(K-t+1)\rangle_K> K-t$,
        which implies that $\mathbf{P}((i_2,j_2),k_1)=*$ from \eqref{eq_star}. If $1\leq K+j_2-k_2\leq K-t$ holds, from \eqref{j2k1_2even} we have
         $\langle j_2-k_1\rangle_K=\langle K+j_2-k_2+\frac{i_2-i_1}{2}(K-t+1)\rangle_K> K-t$ similarly, which also implies that $\mathbf{P}((i_2,j_2),k_1)=*$ from \eqref{eq_star}.

      \item If $\langle K\rangle_{K-t+1}\neq K-t$, we have \eqref{eq_g_even}. If $1\leq j_2-k_2\leq K-t$ holds, we have
        \begin{align*}
        K-t+2&\leq j_2-k_2+\frac{i_2-i_1}{2}(K-t+1)\leq K-t+\frac{g_{new}-2}{2}(K-t+1) \\
        &\overset{\eqref{eq_g_even}}{=}\left\lfloor\frac{K}{K-t+1}\right\rfloor(K-t+1)-1<K,
        \end{align*}
        since $i_1,i_2\in[1:g_{new}]$ and $i_1, i_2$ are even. Therefore, from \eqref{j2k1_2even} we have
        $\langle j_2-k_1\rangle_K=\langle j_2-k_2+\frac{i_2-i_1}{2}(K-t+1)\rangle_K> K-t,$
        which implies that $\mathbf{P}((i_2,j_2),k_1)=*$ from \eqref{eq_star}. If $1\leq K+j_2-k_2\leq K-t$ holds, from \eqref{j2k1_2even} we have
         $\langle j_2-k_1\rangle_K=\langle K+j_2-k_2+\frac{i_2-i_1}{2}(K-t+1)\rangle_K> K-t$ similarly, which also implies that $\mathbf{P}((i_2,j_2),k_1)=*$ from \eqref{eq_star}.
      \end{itemize}

    \item If $i_1$ is odd and $i_2$ is even, we have $s_2=\left\langle k_1-\frac{i_1-1}{2}(K-t+1)\right\rangle_K=\left\langle j_2-\frac{i_2}{2}(K-t+1)\right\rangle_K$
    from \eqref{eq_con}, which implies
    \begin{equation}
    \label{j2k1_odd_even}
    \langle j_2-k_1\rangle_K=\left\langle \frac{i_2-i_1+1}{2}(K-t+1)\right\rangle_K.
    \end{equation}
    \begin{itemize}
    \item If $\langle K\rangle_{K-t+1}= K-t$, we have \eqref{eq_g_odd}, then we have
    \begin{align*}
    K-t+1&\leq\frac{i_2-i_1+1}{2}(K-t+1)\leq \frac{(g_{new}-1)-1+1}{2}(K-t+1) \\ &\overset{\eqref{eq_g_odd}}{=}\left\lfloor\frac{K}{K-t+1}\right\rfloor(K-t+1)\leq K
    \end{align*}
    since $i_1,i_2\in[1:g_{new}]$ and $i_1, i_2$ are odd and even respectively.
    Therefore, from \eqref{j2k1_odd_even} we have $\langle j_2-k_1\rangle_K >K-t$, which implies that $\mathbf{P}((i_2,j_2),k_1)=*$ from \eqref{eq_star}.

    \item If $\langle K\rangle_{K-t+1}\neq K-t$, we have \eqref{eq_g_even}. Then we have
    \begin{align*}
    K-t+1&\leq\frac{i_2-i_1+1}{2}(K-t+1)\leq \frac{g_{new}-1+1}{2}(K-t+1) \\ &\overset{\eqref{eq_g_even}}{=}\left\lfloor\frac{K}{K-t+1}\right\rfloor(K-t+1)\leq K
    \end{align*}
    since $i_1,i_2\in[1:g_{new}]$ and $i_1, i_2$ are odd and even respectively.
    Therefore, from \eqref{j2k1_odd_even} we have $\langle j_2-k_1\rangle_K >K-t$, which implies that $\mathbf{P}((i_2,j_2),k_1)=*$ from \eqref{eq_star}.
    \end{itemize}

    \item If $i_1$ is even and $i_2$ is odd, we have
    \begin{align}
    &s_1=K-t+1-\langle j_1-k_1\rangle_K=\langle j_2-k_2\rangle_K\label{s1_even_odd}\\
    &s_2=\left\langle j_1-\frac{i_1}{2}(K-t+1)\right\rangle_K=\left\langle k_2-\frac{i_2-1}{2}(K-t+1)\right\rangle_K \label{s2_even_odd}
    \end{align}
    from \eqref{eq_con}. From \eqref{s1_even_odd} we have
    \begin{equation}
    \label{j1_even_odd}
    j_1=\langle k_1+K-t+1-(j_2-k_2)\rangle_K.
    \end{equation}
    By Substituting \eqref{j1_even_odd} into \eqref{s2_even_odd}, we can obtain \eqref{j2k1_odd_even}.

    \begin{itemize}
    \item If $\langle K\rangle_{K-t+1}= K-t$, we have \eqref{eq_g_odd}. Then we have
    \begin{align*}
    K-t+1&\leq\frac{i_2-i_1+1}{2}(K-t+1)\leq \frac{g_{new}-2+1}{2}(K-t+1) \\ &\overset{\eqref{eq_g_odd}}{=}\left\lfloor\frac{K}{K-t+1}\right\rfloor(K-t+1)\leq K
    \end{align*}
    since $i_1,i_2\in[1:g_{new}]$ and $i_1, i_2$ are even and odd respectively.
    Therefore, from \eqref{j2k1_odd_even} we have $\langle j_2-k_1\rangle_K >K-t$, which implies that $\mathbf{P}((i_2,j_2),k_1)=*$ from \eqref{eq_star}.

    \item If $\langle K\rangle_{K-t+1}\neq K-t$, we have \eqref{eq_g_even}. Then we have
    \begin{align*}
    K-t+1&\leq\frac{i_2-i_1+1}{2}(K-t+1)\leq \frac{(g_{new}-1)-2+1}{2}(K-t+1) \\ &\overset{\eqref{eq_g_even}}{=}\left\lfloor\frac{K}{K-t+1}\right\rfloor(K-t+1)-(K-t+1)< K
    \end{align*}
    since $i_1,i_2\in[1:g_{new}]$ and $i_1, i_2$ are even and odd respectively.
    Therefore, from \eqref{j2k1_odd_even} we have $\langle j_2-k_1\rangle_K >K-t$, which implies that $\mathbf{P}((i_2,j_2),k_1)=*$ from \eqref{eq_star}.
    \end{itemize}

    \end{itemize}
    Similarly, we can prove that $\mathbf{P}((i_1,j_1),k_2)=*$. Condition C3 of Definition \ref{def-PDA} holds.
The proof is complete.
\end{proof}

\bibliographystyle{IEEEtran}
\bibliography{reference}

\end{document}